\documentclass[aip,10pt,twocolumn]{revtex4-1}

\usepackage{subfigure}
\usepackage{graphicx}
\usepackage{bm}
\usepackage{amssymb}
\usepackage{amsfonts}
\usepackage{amsmath}
\usepackage{lineno}
\usepackage{amsthm}

\newtheorem{proposition}{Proposition}[section]
\newtheorem{remark}{Remark}[section]

\def\FISHER{\mathbf{F}_{\mathcal{R}}}
\def\FISHERR{\mathbf{F_{\mathcal{H}}}}

\def\st{\{\sigma_t\}_{t\ge 0}}
\def\ast{\{\widetilde\sigma_t\}_{t\ge 0}}
\def\PROCS{\st} 
\def\PROCSAPP{\ast} 

\def\LAWEXACT{Q_{[0,T]}^{\theta}}
\def\LAWAPPROX{Q_{[0,T]}^{\theta+\epsilon}}
\def\PATHS{\LAWEXACT}
\def\PATHSAPP{\LAWAPPROX}
\def\EQUIL{\mu^\theta}
\def\EQUILAPP{\mu^{\theta+\epsilon}}

\def\RELENTR#1#2{\mathcal{H}\left({#1}\SEP{#2}\right)}
\def\RELENT#1#2{\mathcal{R}\left({#1}\SEP{#2}\right)}

\def\LATT{{\Lambda}_N}

\def\SPINSP{\Sigma}

\def\R{\mathbb{R}}
\def\NUMSP{{K}}                       
\def\PROCMICRO{\sigma_t}
\def\SEP{{\,|\,}}           
\def\SIGMA{{\mathcal{S}_N}}
\def\SPINSP{\Sigma}
\def\COMMA{\,,}             
\def\PERIOD{\,.}            
\def\Oo{\mathcal{O}}
\def\BIGO{\Oo}

\def\VIZ#1{(\ref{#1})}      

\begin{document}
\title{A Relative Entropy Rate Method for Path Space Sensitivity Analysis
of Stationary Complex Stochastic  Dynamics}

\author{Yannis Pantazis}
\affiliation{Department of Mathematics and Statistics, University of Massachusetts, Amherst, MA, USA.}

\author{Markos A. Katsoulakis}
\affiliation{Department of Mathematics and Statistics, University of Massachusetts, Amherst, MA, USA.}

\date{\today}

\begin{abstract}
We propose a new  sensitivity analysis methodology for complex stochastic dynamics based on the
Relative Entropy Rate. The method  becomes computationally 
feasible  at the stationary  regime of the process and involves the calculation of suitable observables
in path space for the Relative Entropy Rate and the corresponding Fisher Information Matrix. 
The stationary regime is crucial for stochastic dynamics and  here allows us to address the
sensitivity analysis of  complex systems, including examples of processes with complex
landscapes that exhibit metastability, non-reversible systems from a statistical mechanics
perspective, and  high-dimensional, spatially distributed models. All these systems exhibit,
typically  non-gaussian   stationary probability distributions, while in the case of high-dimensionality,
histograms  are impossible to construct directly. Our proposed methods bypass these
challenges relying  on the direct Monte Carlo simulation of rigorously derived  observables for 
the Relative Entropy Rate and Fisher Information in path space rather than on the  stationary
probability distribution itself. We demonstrate the capabilities of the proposed methodology
by focusing here on  two classes of problems:  (a) Langevin particle systems with either reversible
(gradient) or non-reversible (non-gradient) forcing,  highlighting the ability of the  method to carry
out sensitivity analysis in non-equilibrium systems; and, (b) spatially extended  Kinetic Monte Carlo
models,  showing that the method can  handle high-dimensional problems.
\end{abstract}


\maketitle

\section{Introduction}
\label{intro}
In this paper  we propose  the Relative Entropy Rate as a sensitivity analysis  tool  for complex
stochastic dynamics, based on information theory and non-equilibrium statistical mechanics
methods. These calculations become computationally feasible  at the stationary process regime
and involve the calculation of   suitable  observables  in path space for the Relative Entropy
Rate and the corresponding Fisher Information Matrix.  The stationary regime, i.e. stochastic
dynamics where the initial  probability distribution is the stationary distribution reached after long-time integration, 
is especially crucial  for complex systems: it includes dynamic transitions between
metastable states in complex, high-dimensional energy landscapes, intermittency, as well as Non Equilibrium
Steady States (NESS) for non-reversible systems, while at this regime we also construct
phase diagrams for complex systems. Hence their  sensitivity analysis is a crucial question
in determining which parameter directions are the most/least sensitive to perturbations,
uncertainty  or errors resulting from parameter estimation.  Recently there has been
significant progress in developing sensitivity analysis tools for low-dimensional stochastic
processes at the transient regime, such as well-mixed chemical reactions. Some of the
mathematical tools included discrete derivatives \cite{Gunawan:05}, Girsanov transformations
\cite{Nakayama:94,Plyasunov:07}, polynomial chaos \cite{Kim:07},  and coupling
of stochastic processes \cite{Rathinam:10}.

On the other hand, it is often the case that we are interested in  the entire probability density
function (PDF), which in nonlinear and/or discrete systems is typically non-Gaussian, and not
only in a few moments, due to the significance of rare/tail events.  For example, it was recently
shown that in catalytic reactions the most kinetically relevant configurations are occurring rarely,
and correspond to overlapping tails of (non-Gaussian) PDFs \cite{Wu:12}. In that latter direction,
there is a broad recent  literature relying on information theory tools,  where  sensitivity is
estimated by using the  Relative Entropy and the   Fisher Information between PDFs, see for instance
\cite{Liu:06,Ludtke:08,Majda:10,Majda:11,Komorowski:11}. In particular,  such methods  were introduced for the 
study of the sensitivity of PDFs to parameters in climate models \cite{Majda:10}; there the PDFs
structure is known  as it is obtained through an entropy maximization subject to constraints.
Knowing the form of the PDF allows to carry out calculations such as obtaining a Fisher
Information  Matrix (FIM), which in turn identifies the most sensitive parameter directions.
On the other hand,   the sensitivity of stochastic dynamics  can be studied  by using the
FIM \cite{Komorowski:11}.  There the authors  are employing a linearization of
the stochastic evolution around the nonlinear mean field equation and as   a result the
form of the PDF is again known, and more precisely it is Gaussian hence the FIM can
be directly computed. Although there are  regimes where this approximation is applicable
(short times, systems with a single steady state, etc.),  for systems with nontrivial long-time
dynamics, e.g. metastable, it is not correct as  large deviation arguments \cite{Doering:07} show,
or even explicitly available formulas for escape times \cite{Hanggi:84}. Similar
issues  with non-gaussianity in the long time dynamics arise in stochastic systems
with strongly intermittent behavior \cite{KMS}.

Some of these challenges will be addressed through the proposed methods which we
present next in the context of kinetic Monte Carlo  models although similar challenges
 and ideas are relevant to all other stochastic molecular simulation methods. For example,
 we discuss in Section~\ref{num:sec}.C the sensitivity of algorithms for the numerical
 integration of Langevin dynamics. Moreover, kinetic Monte Carlo methods
involving surface chemistry are formulated in terms of continuous time  Markov chains
(jump processes) on a spatial lattice domain  $\LATT$:   at each lattice site $x\in\LATT$
there is a state space  $\SPINSP=\{0,1,\dots, \NUMSP \}$ describing different chemical
species (interacting particles),  where the simplest case $K=1$ represents the well-known
lattice-gas  model \cite{Liggett:85}. The process $\PROCMICRO$  is defined as  a continuous time
Markov  Chain (CTMC)  on the (high-dimensional) state space $\SIGMA=\SPINSP^{\LATT}$
and  mathematically it is    defined  completely by specifying the local  transition rates
$c^\theta(\sigma, \sigma')$ where $\theta\in\R^k$ is a vector of the model parameters.
The transition rates  determine  the  updates  (jumps)  from any current state $\sigma_t=\sigma$
to a (random) new state $\sigma'$ and  concrete examples  of  spatial physicochemical
models are considered in Section~\ref{num:sec}.D. From the local transition rates one defines
the total rate $\lambda^\theta(\sigma)=\sum_{\sigma'} c^\theta(\sigma, \sigma')$, which
is the intensity of the exponential waiting time for a jump from the state $\sigma$. 
The transition probabilities 
are 
$p^\theta(\sigma, \sigma')=\frac{c^\theta(\sigma, \sigma')}{\lambda^\theta(\sigma)}$.
The basic simulation tool for these lattice jump processes is  kinetic Monte Carlo (KMC)
with a wide range of applications   from crystal growth, to catalysis, to biology,  see for
instance \cite{Chatterjee:07}.

\medskip
\noindent
{\em Relative Entropy Rate\,}:
In simulations of dynamic transitions between metastable states on high-dimensional energy
landscapes or of NESS  we are interested in the sensitivity of  stationary processes, i.e., processes
for which the initial probability distribution  is the stationary one (reached after long-time integration).
Mathematically,  we want to assess the sensitivity  of the CTMC 
$\PROCS$  with local transition rates $c^\theta(\sigma, \sigma')$ to a perturbation $\epsilon\in\R^k$
in the parameter vector $\theta$  giving rise to a process $\PROCSAPP$ with local transition rates 
$c^{\theta+\epsilon}(\sigma, \sigma')$, when the initial data are sampled from the respective stationary
probability distribution. The error analysis in the context of the long-time behavior is developed in
terms of the {\em relative entropy},  
\begin{equation}
\label{relent0}
\RELENT{{\PATHS}}{{\PATHSAPP}}=
\int \log\left(\frac{d{\PATHS}}{d{\PATHSAPP}}\right)\, d\PATHS\, ,
\end{equation}
where $\PATHS$ (resp. $\PATHSAPP$) is the {\em path space  probability measures}
of $\PROCS$ (resp. $\PROCSAPP$) in the time interval $[0, T]$. In the case these
probability measures have corresponding probability densities $q^\theta$ and
$q^{\theta+\epsilon}$, \VIZ{relent0} becomes 
 $\RELENT{{\PATHS}}{{\PATHSAPP}}=
\int q^\theta\log\left(\frac{q^\theta}{q^{\theta+\epsilon}}\right)$.
 A key property of the   relative entropy $\RELENT{P}{Q}$ 
is that $\RELENT{P}{Q} \ge 0$ with equality if and only if $P=Q$, 
which  allows us to view relative entropy  as a ``distance" (more precisely a semi-metric)  
between two probability measures $P$ and $Q$. Moreover, from an information theory
perspective \cite{Cover:91},  the relative entropy measures {\it loss/change of information},
e.g. in our context  for  the process $\PROCS$ associated with the parameter vector $\theta$,
with respect to the process $\PROCSAPP$ associated with the parameter vector
$\theta+\epsilon$. Relative entropy for high-dimensional systems was used as measure
of loss of information in coarse-graining\cite{Kats:Trashorras:06, Kats:ReyBellet:07,
Arnst:08}, and sensitivity analysis for climate modeling problems \cite{Majda:10}.

Starting from \VIZ{relent0}, by Girsanov's formula we obtain an explicit expression for the
corresponding Radon-Nikodym derivative 
{\small
\begin{equation}
\begin{aligned}\label{RN1}
\frac{d\PATHS}{d\PATHSAPP}  (\{\sigma_t\}) &=
\exp \Big\{ \sum_{s \leq T} \log \frac{\lambda^\theta(\sigma_{s-}) p^\theta(\sigma_{s-},\sigma_s)}{\lambda^{\theta+\epsilon}(\sigma_{s-}) 
p^{\theta+\epsilon}(\sigma_{s-},\sigma_s)} \\
&- \int_{0}^{T} [\lambda^\theta(\sigma_s) - \lambda^{\theta+\epsilon}(\sigma_s)]\,ds \Big\}\COMMA
\end{aligned}
\end{equation}
}\hspace{-2mm} on any path of the process $\{\sigma_t\}_{t \in [0,T]}$ in terms of the jump rates and
transition probabilities of both process, under suitable non-degeneracy conditions
\cite{Kipnis:99}. Notice that $\sigma_{s-}$ denotes the left-hand limit of $\sigma_s$ at a jump instance $s$.
Following   calculations regarding the related quantity of entropy production in non-equilibrium 
statistical mechanics \cite{Maes:00}, 
we can show  that when the initial distribution
$\sigma_0 \sim \EQUIL$ where $\EQUIL$ (resp. $\EQUILAPP$) is the
stationary probability disturbution   of $\PROCS$ (resp. $\PROCSAPP$), then the relative entropy
formula simplifies  dramatically in two parts, one pure equilibrium (scaling as
$\BIGO(1)$) and one  capturing the stationary dynamics (scaling as $\BIGO(T)$):
{\small
\begin{equation}\label{relent1}
   \RELENT{{\PATHS}}{{\PATHSAPP}}
    =T \RELENTR{{\PATHS}}{{\PATHSAPP}}+\RELENT{\EQUIL}{\EQUILAPP},
\end{equation}
}\hspace{-2mm} where $\RELENT{\EQUIL}{\EQUILAPP}$  is the relative entropy between
the stationary probabilities, while  
{\small
\begin{equation}
\begin{aligned}\label{relent2}
&\RELENTR{{\PATHS}}{{\PATHSAPP}} = \mathbb{E}_{\mu^\theta} \Big[
\sum_{\sigma'} \lambda^\theta(\sigma) p^\theta(\sigma, \sigma') \\
&\times \log \frac{\lambda^\theta(\sigma) p^\theta(\sigma, \sigma')}{\lambda^{\theta+\epsilon}(\sigma) p^{\theta+\epsilon}(\sigma, \sigma')}
- (\lambda^\theta(\sigma) - \lambda^{\theta+\epsilon}(\sigma))\Big]\COMMA
\end{aligned}
\end{equation}
}\hspace{-2mm} where $\mathbb{E}_{\mu^\theta}$ denotes the expected value with respect
to the probability $\mu^\theta$. In \VIZ{relent1}, we immediately notice that 
a most relevant quantity to describe the change  of information content upon perturbation
of model parameters  of a stochastic process is the $\BIGO(T)$ term, which can be thought
as a {\em relative entropy per  unit time} while on the other hand,  the term
$\RELENT{\EQUIL}{\EQUILAPP}$ becomes unimportant as $T$ grows.

We will refer from now on to the quantity \VIZ{relent2} as the {\em Relative Entropy Rate}
(RER), which can be thought as the change in information per unit time. Notice that RER
has the correct time scaling since it is actually independent of the interval $[0,T]$. Furthermore,
\VIZ{relent2} provides a computable observable that can be sampled from the steady
state $\mu^\theta$ in terms of conventional  KMC,  bypassing  the need for a histogram
or an explicit formula for the high-dimensional probabilities involved in \VIZ{relent0}.
Finally, the fact that in stationary  regimes,  when $T\gg 1$  in \VIZ{relent1},
the term $\RELENT{\mu^\theta}{\mu^{\theta+\epsilon}}$ becomes unimportant, is
especially convenient: $\mu^\theta$ and $\mu^{\theta+\epsilon}$ are typically not known
explicitly in non-reversible systems, for instance  in  spatially distributed reaction  KMC
or non-reversible Langevin dynamics considered here as  examples.

\medskip
\noindent
{\em Fisher Information Matrix on Path Space\,}:
 An attractive approach to sensitivity analysis that is rigorously based on relative entropy
 calculations  is the Fisher Information Matrix. Indeed, assuming smoothness in the parameter
 vector, it is straightforward to obtain the expansion of \VIZ{relent0} \cite{Cover:91, Abramov:05},
\begin{equation}\label{GFIM0}
\RELENT{{\PATHS}}{{\PATHSAPP}}
= \frac{1}{2} \epsilon^T \FISHER(\PATHS) \epsilon + O(|\epsilon|^3)\, , 
\end{equation}
where the Fisher Information Matrix (FIM) is defined as the Hessian of the relative entropy:
\begin{equation}\label{Fisher}
  \FISHER(\PATHS) =  \left. \nabla_\epsilon^2 \RELENT{{\PATHS}}{{\PATHSAPP}}\right|_{\epsilon=0}
  \PERIOD
\end{equation}
As \VIZ{GFIM0} readily suggests, relative entropy is locally a quadratic function of the parameter
vector $\theta$. Thus spectral analysis of $\FISHER$--provided the matrix is available--would
allow us to identify  which parameter directions are the most/least sensitive to perturbations,
uncertainty  or errors resulting from parameter estimation. The source of such uncertainties could
be related to the assimilation of  experimental data \cite{Rao:10} or  finer scale
numerical simulation, e.g. Density Functional Theory computations  in the case of molecular
simulations \cite{Braatz:06}. More specifically, the knowledge of the Fisher Information Matrix not
only provides a gradient-free method for sensitivity analysis, but  allows   to address questions
of parameter identifiability \cite{Rothenberg:71,Komorowski:11} and optimal experiment
design \cite{Emery:98,Prasad:10}.
However, the FIM  $\FISHER$ in \VIZ{Fisher}  is not accessible computationally in general,
nevertheless analytic calculations can be performed at equilibrium (e.g., in ergodic systems
when $T \to\infty$) under the assumption or the explicit knowledge of the stationary distribution
$\mu$. An example of such a calculation is under the assumption of a Gaussian
distribution with the mean $m(\theta)$ and the covariance matrix $\Sigma(\theta)$ in which
case the matrix $\FISHER$ is computed in terms of derivatives of the mean and covariance
matrix \cite{Komorowski:11}.

On the other hand  \VIZ{relent1} provides  a different perspective to these issues, giving rise 
 to a computable  observable  for the {\em path space  Fisher Information Matrix}  that includes
 transition rates rather than just the stationary  PDFs. Indeed, by combining \VIZ{relent1} and
 \VIZ{GFIM0} we  obtain the following expansion for the dominant, $O(T)$  term in \VIZ{relent1}:
\begin{equation}\label{GFIM}
\RELENTR{\PATHS}{\PATHSAPP} = \frac{1}{2} \epsilon^T \FISHERR(\PATHS) \epsilon + O(|\epsilon|^3)\, , 
\end{equation}
where  the  {\em Fisher Information Matrix} per unit time, $\FISHERR(\PATHS)$,
has the explicit form
\begin{equation}
\begin{aligned}\label{GFIM2}
\FISHERR(\PATHS) &= \mathbb{E}_{\mu^\theta} \Big[ \sum_{\sigma'}c^\theta(\sigma, \sigma') \\
&\times \nabla_\theta \log c^\theta(\sigma, \sigma') \nabla_\theta \log c^\theta(\sigma, \sigma')^T\Big]\, ,
\end{aligned}
\end{equation}
where $c^\theta(\sigma, \sigma')=\lambda^\theta(\sigma)p^\theta(\sigma, \sigma')$.
Fisher Information Matrices given by \VIZ{Fisher} and \VIZ{GFIM2} are straightforwardly
related through $\lim_{T\rightarrow\infty} \frac{1}{T} \FISHER = \FISHERR$.
It is clear from \VIZ{GFIM2} that the Fisher Information Matrix, just like the Relative Entropy
Rate \VIZ{relent2},   is merely an observable that can be sampled using  KMC algorithms.

The previous discussion suggests that  the proposed approach to sensitivity analysis is expected to have  the following features: 
\begin{enumerate}
\item It is rigorously valid for the sensitivity of  long-time, stationary  dynamics in path space, including
for example metastable dynamics in a complex landscape.

\item It is a gradient-free sensitivity analysis method which does not require the knowledge
of the equilibrium PDFs, as \VIZ{Fisher} is replaced with a computable observable
\VIZ{GFIM2}, that contains explicitly information for local dynamics.

\item It is suitable for non-equilibrium systems  from a statistical mechanics perspective; for example,
non-reversible processes, such as spatially extended reaction-diffusion Kinetic Monte
Carlo, where the structure of the equilibrium PDF is unknown and is typically non-Gaussian.

\item A key enabling tool for implementing the proposed methodology in high-dimensional
stochastic systems is molecular simulation methods such as KMC or Langevin solvers which
can sample the observables \VIZ{relent2} and \VIZ{GFIM2}, and in particular  their accelerated
or scalable  versions \cite{Gillespie:01, Chatterjee:07, SPPARKS,  LAMMPS, Arampatzis:12}.
\end{enumerate}
Indeed, we demonstrate these features by presenting three examples addressing different points: (a) the well-mixed bistable
reaction system known as the Schl\"ogl model which also serves as a benchmark;  (b) a Langevin particle system with either reversible
or non-reversible forcing, that demonstrates the ability of the proposed method to carry out sensitivity analysis in non-equilibrium systems; and,
 (c) a spatially extended  KMC model for $CO$ oxidation known
as the Ziff-Gulari-Barshad (ZGB) model. Such reaction-diffusion models are typically
non-reversible, hence the sensitivity tools we propose here are highly suitable. Regarding this last class of problems, we note 
that in  more accurate, state-of-the-art    KMC models  with a large number of parameters \cite{Hansen:00,Meskine:09,Stamatakis:11},
kinetic parameters are estimated through density functional theory (DFT) calculations,
hence sensitivity analysis is a crucial step in determining the parameters that need to be
calculated with greater accuracy.

The paper is organized as follows: in Section~\ref{Markov:chain} we present the derivation
of the Relative Entropy Rate and its corresponding Fisher Information Matrix for discrete-time
Markov chains while Section~\ref{Markov:processes} the same observables for continuous-time
Markov processes (i.e., \VIZ{relent1}, \VIZ{relent2} and \VIZ{GFIM2}) are derived.
Section~\ref{Markov:generalizations} generalizes the RER and the FIM to time-periodic,
inhomogeneous Markov processes as well as to semi-Markov processes. Statistical estimators and numerical
examples in Section~\ref{num:sec} demonstrate the efficiency of the proposed sensitivity
method, while Section~\ref{concl} concludes the paper.

\section{Discrete Time Markov Chains}
\label{Markov:chain}
Let $\{\sigma_m\}_{m\in\mathbb Z^+}$ be a discrete-time time-homogeneous Markov chain with
separable state space $E$. The transition probability kernel of the Markov chain denoted
by  $P^\theta(\sigma, d\sigma')$ depends on the parameter vector $\theta\in\mathbb R^k$.
Assume that the transition kernel is absolute continuous with respect to (w.r.t.) the Lebesgue
measure\cite{Note1} and the transition probability density function $p^\theta(\sigma,\sigma')$
is always positive for all $\sigma,\sigma'\in E$ and for all $\theta\in\mathbb R^k$. We further assume that
$\{\sigma_m\}_{m\in\mathbb Z^+}$ has a unique stationary probability distribution  denoted
by $\mu^\theta(\sigma)$.
Exploiting the Markov
property, the path probability distribution $Q_{0,M}^{\theta}$ for the path $\{\sigma_m\}_{m=0}^M$
at the time horizon $0,...,M$ starting  from the stationary distribution 
 $\mu^\theta(\sigma_0)$
is given by
\begin{equation}
Q_{0,M}^{\theta}\big(\sigma_0,\cdot\cdot\cdot, \sigma_M \big) = \mu^\theta(\sigma_0)
p^\theta(\sigma_0,\sigma_1) \cdot\cdot\cdot p^\theta(\sigma_{M-1},\sigma_M)\, .
\end{equation}
We consider the perturbation by $\epsilon\in\mathbb R^k$ and the Markov chain
$\{\widetilde{\sigma}_m\}_{m\in\mathbb Z^+}$ with the respective transition probability
density function, $p^{\theta+\epsilon}(\sigma,\sigma')$, the
respective stationary density, $\mu^{\theta+\epsilon}(\sigma)$, as well as the respective
path distribution $Q_{0,M}^{\theta+\epsilon}$. Then, the Radon-Nikodym derivative
of the unperturbed path distribution w.r.t. the perturbed path distribution takes the form
\begin{equation} \label{RNderiv:chain}
\frac{dQ_{0,M}^\theta}{d Q_{0,M}^{\theta+\epsilon}}\big( \{\sigma_m\} \big) =
\frac{\mu^{\theta}(\sigma_0)\prod_{i=0}^{M-1}p^{\theta}(\sigma_i,\sigma_{i+1})}
{\mu^{\theta+\epsilon}(\sigma_0)\prod_{i=0}^{M-1}p^{\theta+\epsilon}(\sigma_i,\sigma_{i+1})}\, ,
\end{equation}
which is well-defined since the transition probabilities are assumed always positive.

The following Proposition demonstrates the relative entropy representation of
the path distribution $Q_{0,M}^{\theta}$ w.r.t. the path distribution
$Q_{0,M}^{\theta+\epsilon}$. 

\begin{proposition}\label{propos:MC}
Under the previous assumptions, the path space relative entropy $\RELENT{Q_{0,M}^{\theta}}{Q_{0,M}^{\theta+\epsilon}} :=
\int \log\left(\frac{d{Q_{0,M}^{\theta}}}{d{Q_{0,M}^{\theta+\epsilon}}}\right) dQ_{0,M}^{\theta}$
equals to
\begin{equation}
\RELENT{Q_{0,M}^{\theta}}{Q_{0,M}^{\theta+\epsilon}} = M \RELENTR{Q_{0,M}^{\theta}}{Q_{0,M}^{\theta+\epsilon}}
+ \RELENT{\mu^\theta}{\mu^{\theta+\epsilon}}
\end{equation}
where
\begin{equation} \label{RER:MC}
\RELENTR{Q_{0,M}^{\theta}}{Q_{0,M}^{\theta+\epsilon}}
= \mathbb E_{\mu^\theta} \left[ \int_E p^{\theta}(\sigma,\sigma')
\log\frac{p^{\theta}(\sigma,\sigma')}{p^{\theta+\epsilon}(\sigma,\sigma')} d\,\sigma' \right]
\end{equation}
is the relative entropy rate.
\end{proposition}

\begin{proof}
The path space relative entropy equals to {\small
\begin{equation*}
\begin{aligned}
&\RELENT{{Q_{0,M}^{\theta}}}{{Q_{0,M}^{\theta+\epsilon}}} 
= \int_{E}\cdot\cdot\cdot\int_{E} \mu^{\theta}(\sigma_0)\prod_{i=0}^{M-1}p^{\theta}(\sigma_i,\sigma_{i+1}) \\
&\times\log \frac{\mu^{\theta}(\sigma_0)\prod_{i=0}^{M-1}p^{\theta}(\sigma_i,\sigma_{i+1})}
{\mu^{\theta+\epsilon}(\sigma_0)\prod_{i=0}^{M-1}p^{\theta+\epsilon}(\sigma_i,\sigma_{i+1})}
d\sigma_0\cdot\cdot\cdot d\sigma_M \\
&= \int_{E}\cdot\cdot\cdot\int_{E} \mu^{\theta}(\sigma_0)\prod_{i=0}^{M-1}p^{\theta}(\sigma_i,\sigma_{i+1})
 \left( \log\frac{\mu^{\theta}(\sigma_0)}{\mu^{\theta+\epsilon}(\sigma_0)} \right. \\ & \left.
+ \sum_{i=0}^{M-1} \log\frac{p^{\theta}(\sigma_i,\sigma_{i+1})}{p^{\theta+\epsilon}(\sigma_i,\sigma_{i+1})}
\right)  d\sigma_0\cdot\cdot\cdot d\sigma_M
\end{aligned}
\end{equation*}}

Using the relations {\small
\begin{equation*}
\int_E p(\sigma,\sigma') d\sigma' = 1
\ \ \ \ \ \ \ \ \ \ \& \ \ \ \ \ \ \ \ \ \ 
\int_E \mu(\sigma) p(\sigma,\sigma') d\sigma = \mu(\sigma')
\end{equation*}
}the relative entropy is simplified to {\small
\begin{equation*}
\begin{aligned}
&\RELENT{{Q_{0,M}^{\theta}}}{{Q_{0,M}^{\theta+\epsilon}}} = 
\int_{E} \mu^{\theta}(\sigma_0) \log\frac{\mu^{\theta}(\sigma_0)}{\mu^{\theta+\epsilon}(\sigma_0)} d\sigma_0 \\
&+ \sum_{i=0}^{M-1} \int_{E}\int_{E} \mu^{\theta}(\sigma_i)p^{\theta}(\sigma_i,\sigma_{i+1})
\log\frac{p^{\theta}(\sigma_i,\sigma_{i+1})}{p^{\theta+\epsilon}(\sigma_i,\sigma_{i+1})} d\sigma_i d\sigma_{i+1} \\
&= M \RELENTR{{Q_{0,M}^{\theta}}}{{Q_{0,M}^{\theta+\epsilon}}} + \RELENT{\mu^{\theta}}{\mu^{\theta+\epsilon}}
\end{aligned}
\end{equation*}
}\end{proof}

For large times ($M\gg 1$), the significant term of the relative entropy,
$\RELENT{{Q_{0,M}^{\theta}}}{{Q_{0,M}^{\theta+\epsilon}}}$,
is the relative entropy rate, $\RELENTR{{Q_{0,M}^{\theta}}}{{Q_{0,M}^{\theta+\epsilon}}}$,
which scales linearly with the number of jumps of the Markov chain
while the relative entropy between the stationary probability distributions, $\RELENT{\mu^\theta}{\mu^{\theta+\epsilon}}$,
becomes unimportant.
Thus, at the stationary regime, the appropriate observable for sensitivity analysis
is the relative entropy rate. Furthermore, the RER expression \VIZ{RER:MC} incorporates
the transition probabilities of the Markov chain which are typically known --for instance,
whenever a path sample is needed to be generated-- while the respective stationary probability distributions
are typically unknown --for instance, in non-reversible systems-- and should
be computed numerically, if possible. Moreover,
the  path-space RER takes into consideration the dynamical aspects
of the process while the relative entropy between the stationary distributions does
not take into account any dynamical aspects of the process which are critical in
metastable or intermittent regimes.

\medskip
\noindent
{\em Fisher Information Matrix for Relative entropy rate\,}: The relative entropy rate
is locally a quadratic functional in a neighborhood of $\theta$. The curvature of
the RER around $\theta$, defined by its Hessian, is called the Fisher Information
Matrix which is formally derived as follows. 
Let $\delta p(\sigma,\sigma') := p^{\theta+\epsilon}(\sigma,\sigma') - p^\theta(\sigma,\sigma')$,
then the relative entropy rate of $Q_{0,M}^{\theta}$ w.r.t. $Q_{0,M}^{\theta+\epsilon}$
is written as {\small
\begin{equation*}
\begin{aligned}
&\RELENTR{{Q_{0,M}^{\theta}}}{{Q_{0,M}^{\theta+\epsilon}}} \\
&= - \int_E\int_E \mu^\theta(\sigma) p^\theta(\sigma,\sigma') \log \left(1+ \frac{\delta p(\sigma,\sigma')}{p^\theta(\sigma,\sigma')}\right) d\sigma d\sigma' \\
&= - \int_E\int_E \left[\mu^\theta(\sigma)\delta p(\sigma,\sigma') \right. \\
&\left.-  \frac{1}{2}\mu^\theta(\sigma)\frac{\delta p(\sigma,\sigma')^2}{p^\theta(\sigma,\sigma')}
+ O(|\delta p(\sigma,\sigma')|^3) \right] d\sigma d\sigma' \, .
\end{aligned}
\end{equation*}
} Moreover, for all $\sigma\in E$, it holds that {\small
\begin{equation*}
\int_E \delta p(\sigma,\sigma')d\sigma' = \int_E p^{\theta+\epsilon}(\sigma,\sigma')d\sigma'
- \int_E p^{\theta}(\sigma,\sigma')d\sigma' = 1-1 = 0
\end{equation*}
} \hspace{-3mm} while a under smoothness assumption on the transition probability function
for the parameter $\theta$, which is an easily checkable assumption, a
Taylor series expansion is applicable to $\delta p$: 
\begin{equation*}
\delta p(\sigma,\sigma') = \epsilon^T \nabla_\theta p^\theta(\sigma,\sigma') + O(|\epsilon|^2)
\end{equation*}
Thus, we finally obtain that {\small
\begin{equation*}
\begin{aligned}
&\RELENTR{{Q_{0,M}^{\theta}}}{{Q_{0,M}^{\theta+\epsilon}}} \\
&= \frac{1}{2} \int_E\int_E  \mu^\theta(\sigma) \frac{(\epsilon^T \nabla_\theta p^\theta(\sigma,\sigma'))^2}
{p^\theta(\sigma,\sigma')} d\sigma d\sigma' + O(|\epsilon|^3) \\
&= \frac{1}{2}\epsilon^T\Big( \int_E\int_E \mu^\theta(\sigma) p^\theta(\sigma,\sigma) \nabla_\theta \log p^\theta(\sigma,\sigma') \\
&\ \ \ \ \ \ \ \ \ \ \ \times \nabla_\theta \log p^\theta(\sigma,\sigma')^T d\sigma d\sigma' \Big) \epsilon + O(|\epsilon|^3) \\
&= \frac{1}{2} \epsilon^T \FISHERR\big({{Q_{0,M}^{\theta}}}\big) \epsilon + O(|\epsilon|^3)
\end{aligned}
\end{equation*}}
where {\small
\begin{equation}
\begin{aligned}\label{FIM:MC}
&\FISHERR\big({{Q_{0,M}^{\theta}}}\big) := \\
&\mathbb E_{\mu^\theta} \left[\int_E p^\theta(\sigma,\sigma) \nabla_\theta
\log p^\theta(\sigma,\sigma') \nabla_\theta \log p^\theta(\sigma,\sigma')^T d\,\sigma'\right]
\end{aligned}
\end{equation}
} is the path space Fisher Information Matrix (FIM) for the relative entropy rate.
Notice that FIM as well as RER are computed from the transition probabilities
under mild ergodic average assumptions
(see also Section~\ref{num:sec} where explicit numerical formulas are provided).

\begin{remark}
The Fisher information Matrix for $\RELENTR{{Q_{0,M}^{\theta+\epsilon}}}{{Q_{0,M}^{\theta}}}$
is again $\FISHERR\big({{Q_{0,M}^{\theta}}}\big)$ while the relative entropy rates
are related for small $\epsilon$ through {\small
\begin{equation}
\begin{aligned}
&\RELENTR{{Q_{0,M}^{\theta+\epsilon}}}{{Q_{0,M}^{\theta}}} = \RELENTR{{Q_{0,M}^{\theta}}}{{Q_{0,M}^{\theta+\epsilon}}} + O(|\epsilon|^3) \\
&= \RELENTR{{Q_{0,M}^{\theta}}}{{Q_{0,M}^{\theta-\epsilon}}} + O(|\epsilon|^3) \, .
\end{aligned}
\end{equation}}
\end{remark}

\begin{remark}
If the transition probability function of the Markov chain equals to
$p^\theta(\sigma,\sigma') = \mu^\theta(\sigma')$ for all $\sigma,\sigma'\in E$
and for all $\theta\in\mathbb R^k$, which is equivalent to the fact that the samples are
independent, identically distributed from the stationary probability distribution, then the relative entropy
rate between the path probabilities becomes the usual relative entropy between the
stationary distributions and the path space FIM becomes the usual FIM. Indeed, FIM
is simplified to {\small
\begin{equation*}
\begin{aligned}
&\FISHERR\big({{Q_{0,M}^{\theta}}}\big) \\
&= \int_E\int_E \mu^\theta(\sigma) \mu^\theta(\sigma') \nabla_\theta \log \mu^\theta(\sigma') \nabla_\theta \log \mu^\theta(\sigma')^T d\sigma d\sigma' \\
&= \int_E \mu^\theta(\sigma') \nabla_\theta \log \mu^\theta(\sigma') \nabla_\theta \log \mu^\theta(\sigma')^T d\sigma' \\
&=: \FISHER\big(\mu^\theta\big)
\end{aligned}
\end{equation*}
}while we similarly obtain for the relative entropy rate that
$\mathcal H(P_{0t}^\theta | P_{0t}^{\theta+\epsilon}) = \mathcal R(\mu^\theta | \mu^{\theta+\epsilon})$.
\end{remark}

\section{Continuous-time Markov Chains}
\label{Markov:processes}
As in the case of Kinetic Monte Carlo methods, we consider $\{\sigma_t\}_{t\in\mathbb R^+}$
to be a CTMC 
with countable state space $E$. The parameter dependent transition
rates denoted by $c^\theta(\sigma, \sigma')$ completely define the jump Markov process.
The transition rates determine the updates (jumps or
sojourn times) from a current state $\sigma$ to a new (random) state $\sigma'$ through
the total rate $\lambda^\theta(\sigma):=\sum_{\sigma'\in E} c^\theta(\sigma, \sigma')$ which
is the intensity of the exponential waiting time for a jump from state $\sigma$. The
transition probabilities for the embedded Markov chain $\big\{J_n\big\}_{n\geq0}$
are $p^\theta(\sigma, \sigma') = \frac{c^\theta(\sigma, \sigma')}{\lambda^\theta(\sigma)}$
while the generator of the jump Markov process is an operator acting on the bounded
functions (also called observables) $f(\sigma)$ defined on the state space $E$ and fully
determines the process:
\begin{equation}
\mathcal L f(\sigma) = \sum_{\sigma'\in E} c^\theta(\sigma, \sigma')[f(\sigma')-f(\sigma)] \, .
\end{equation}

Assume that a new jump Markov process $\{\widetilde\sigma_t\}_{t\in\mathbb R^+}$ is
defined by perturbing the transition rates by a small vector $\epsilon\in\mathbb R^k$
and that the two path probabilities  $\PATHS$ and $\PATHSAPP$ are absolute continuous
with respect to  each other which is satisfied when $c^{\theta}(\sigma, \sigma')=0$ if and
only if $c^{\theta+\epsilon}(\sigma, \sigma')=0$ holds for all $\sigma,\sigma'\in E$.
Then the Radon-Nikodym derivative of the path distribution $\PATHS$ with respect to  the
path distribution $\PATHSAPP$ has a explicit formula known also as Girsanov formula
\cite{Liptser:77, Kipnis:99} {\small
\begin{equation}
\begin{aligned}\label{RN:MP}
\frac{d\PATHS}{d\PATHSAPP}  (\{\sigma_t\}) &= \frac{\EQUIL(\sigma_0)}{\EQUILAPP(\sigma_0)}
\exp \left\{ \int_0^T \log \frac{c^\theta(\sigma_{s-},\sigma_s)}{c^{\theta+\epsilon}(\sigma_{s-},\sigma_s)}dN_s \right. \\
&\left. - \int_{0}^{T} [\lambda^\theta(\sigma_s) - \lambda^{\theta+\epsilon}(\sigma_s)]\,ds  \right\}\COMMA
\end{aligned}
\end{equation}
}where $\EQUIL$ (reps. $\EQUILAPP$) is the stationary distributions  of $\{\sigma_t\}_{t\in\mathbb R^+}$
(resp. $\{\widetilde\sigma_t\}_{t\in\mathbb R^+}$) while $N_s$ is the counting (of the jumps) measure.
Having the Girsanov formula, the relative entropy is easily derived as the next Proposition reveals.

\begin{proposition}\label{propos:MP}
Under the previous assumptions, the path space relative entropy $\RELENT{{\PATHSAPP}}{{\PATHS}}$
equals to
\begin{equation}\label{relent:MP1}
\RELENT{{\PATHS}}{{\PATHSAPP}} = T \RELENTR{{\PATHS}}{{\PATHSAPP}}+\RELENT{\EQUIL}{\EQUILAPP}\COMMA
\end{equation}
where { \small
\begin{equation}
\begin{aligned}\label{RER:MP}
\RELENTR{{\PATHS}}{{\PATHSAPP}}
= \mathbb E_{\mu^\theta} \Big[& \sum_{\sigma'\in E} c^{\theta}(\sigma, \sigma')
\log \frac{c^\theta(\sigma, \sigma')}{c^{\theta+\epsilon}(\sigma, \sigma')} \\
&-(\lambda^{\theta}(\sigma) - \lambda^{\theta+\epsilon}(\sigma)) \Big]
\end{aligned}
\end{equation}
} is the relative entropy rate.
\end{proposition}

\begin{proof}
The explicit formula for the RER was first given by Dumitrescu \cite{Dumitrescu:88}
for finite state space, though, we reproduce the proof for the sake of completeness.
Using the Girsanov formula, the relative entropy \VIZ{relent:MP1} is rewritten as {\small
\begin{equation*}
\begin{aligned}
&\RELENT{{\PATHS}}{{\PATHSAPP}} \\
&= \mathbb E_{\PATHS} \left[ \log \frac{\mu^{\theta}(\sigma_0)}{\mu^{\theta+\epsilon}(\sigma_0)}
\int_0^T \log \frac{c^\theta(\sigma_{s-},\sigma_s)}{c^{\theta+\epsilon}(\sigma_{s-},\sigma_s)}\,dN_s \right.\\
&\left.- \int_{0}^{T} [\lambda^\theta(\sigma_s) - \lambda^{\theta+\epsilon}(\sigma_s)]\,ds \right] \\
&= \mathbb E_{\PATHS} \left[ \int_0^T \log \frac{c^\theta(\sigma_{s-},\sigma_s)}{c^{\theta+\epsilon}(\sigma_{s-},\sigma_s)}\,dN_s \right] \\
&- \mathbb E_{\PATHS} \left[ \int_{0}^{T} [\lambda^\theta(\sigma_s) - \lambda^{\theta+\epsilon}(\sigma_s)]\,ds \right]
+ \mathbb E_{\PATHS} \left[ \log \frac{\mu^{\theta}(\sigma_0)}{\mu^{\theta+\epsilon}(\sigma_0)} \right]
\end{aligned}
\end{equation*}
}Exploiting the fact that the process $M_t:=N_t-\int_0^t \lambda^{\theta}(\sigma_{s-})d\,s$
is a martingale, we have that {\small
\begin{equation*}
\begin{aligned}
&\mathbb E_{\PATHS} \left[ \int_0^T \log \frac{c^\theta(\sigma_{s-},\sigma_s)}{c^{\theta+\epsilon}(\sigma_{s-},\sigma_s)}\,dN_s \right] \\
&= \mathbb E_{\PATHS} \left[ \int_0^T \lambda^{\theta}(\sigma_{s-})
\log \frac{c^\theta(\sigma_{s-},\sigma_s)}{c^{\theta+\epsilon}(\sigma_{s-},\sigma_s)}\,ds \right] \ .
\end{aligned}
\end{equation*}
}Moreover, changing the order of the integrals and due to the stationarity of the process
$\{\sigma_t\}_{t\in\mathbb R^+}$, the relative entropy is simplified to the following:  {\small
\begin{equation*}
\begin{aligned}
&\RELENT{{\PATHS}}{{\PATHSAPP}} \\
&=\int_{0}^{T} \mathbb E_{\mu^{\theta}}\left[ \sum_{\sigma'\in E} \lambda^{\theta}(\sigma)p^{\theta}(\sigma, \sigma')
\log \frac{c^\theta(\sigma, \sigma')}{c^{\theta+\epsilon}(\sigma, \sigma')}\right]\,ds \\
&- \int_{0}^{T} \mathbb E_{\mu^{\theta}} \left[\lambda^\theta(\sigma) - \lambda^{\theta+\epsilon}(\sigma)\right]\,ds
+ \mathbb E_{\mu^{\theta}} \left[ \log \frac{\mu^{\theta}(\sigma)}{\mu^{\theta+\epsilon}(\sigma)} \right] \\
&= T \RELENTR{{\PATHS}}{{\PATHSAPP}}+\RELENT{\EQUIL}{\EQUILAPP}
\end{aligned}
\end{equation*}
} \end{proof}

\medskip
\noindent
{\em Fisher Information Matrix\,}: Even though not directly evident, relative entropy
rate for the jump Markov processes is locally a quadratic function of the parameter vector
$\theta\in\mathbb R^k$. Hence, Fisher Information Matrix which is defined as the Hessian
of the RER can be derived.
Indeed, defining the rate difference $\delta c(\sigma,\sigma') = c^{\theta+\epsilon}(\sigma,\sigma')
- c^\theta(\sigma,\sigma')$, the relative entropy rate of $\PATHS$ w.r.t. $\PATHSAPP$ equals to {\small
\begin{equation}
\begin{aligned}
&\RELENTR{{\PATHS}}{{\PATHSAPP}} \\
&= - \sum_{\sigma,\sigma'\in E} \mu^\theta(\sigma) c^\theta(\sigma,\sigma')
\log \left(1+ \frac{\delta c(\sigma,\sigma')}{c^\theta(\sigma,\sigma')}\right) \\
&+ \sum_{\sigma,\sigma'\in E} \mu^\theta(\sigma)\delta c(\sigma,\sigma') \\
&= - \sum_{\sigma,\sigma'\in E} \Big[ \mu^\theta(\sigma)\delta c(\sigma,\sigma')
-\frac{1}{2}\mu^\theta(\sigma)\frac{\delta c(\sigma,\sigma')^2}{c^\theta(\sigma,\sigma')} \\
&+ O(|\delta c(\sigma,\sigma')|^3) \Big]
+ \sum_{\sigma,\sigma'\in E} \mu^\theta(\sigma)\delta c(\sigma,\sigma') \\
&= \frac{1}{2}\sum_{\sigma,\sigma'\in E} \mu^\theta(\sigma)\frac{\delta c(\sigma,\sigma')^2}{c^\theta(\sigma,\sigma')}
+ O(|\delta c(\sigma,\sigma')|^3)
\end{aligned}
\end{equation}
} Under a smoothness assumption on the transition rates in a neighborhood of parameter
vector $\theta$, which is also a checkable hypothesis, a Taylor series expansion of
$\delta c(\sigma,\sigma') = \epsilon^T \nabla_\theta c^\theta(\sigma,\sigma') + O(|\epsilon|^2)$
results in {\small
\begin{equation}
\begin{aligned}
&\RELENTR{{\PATHS}}{{\PATHSAPP}} \\
&= \frac{1}{2} \sum_{\sigma,\sigma'\in E} \mu^\theta(\sigma)
\frac{\big(\epsilon^T \nabla_\theta c^\theta(\sigma,\sigma')\big)^2} {c^\theta(\sigma,\sigma')} + O(|\epsilon|^3) \\
&= \frac{1}{2}\epsilon^T\Big( \sum_{\sigma,\sigma'\in E} \mu^\theta(\sigma)
c^\theta(\sigma,\sigma') \nabla_\theta \log c^\theta(\sigma,\sigma') \\
&\ \ \ \ \ \ \ \ \ \ \times\nabla_\theta \log c^\theta(\sigma,\sigma')^T \Big) \epsilon + O(|\epsilon|^3) \\
&= \frac{1}{2} \epsilon^T \FISHERR(\PATHS) \epsilon + O(|\epsilon|^3)
\end{aligned}
\end{equation}
} where  {\small
\begin{equation}
\begin{aligned}\label{FIM:MP}
&\FISHERR(\PATHS) := \\
&\mathbb E_{\mu^{\theta}}\left[ \sum_{\sigma'\in E} c^\theta(\sigma,\sigma')
\nabla_\theta \log c^\theta(\sigma,\sigma') \nabla_\theta \log c^\theta(\sigma,\sigma')^T \right]
\end{aligned}
\end{equation}
} \hspace{-3mm} is the path space Fisher information matrix of a jump Markov process.
It is based on the transition rates of the process which are typically known
---they actually define the process--- thus FIM as well as RER are numerically
computable under mild ergodicity assumptions.
Furthermore, it is noteworthy that the only difference between the FIM of  the
Markov chains in the previous Section and the FIM of the continuous-time jump
Markov processes is that in the latter the transition rates $c^\theta(\sigma,\sigma')$
are employed instead of the transition probabilities $p^\theta(\sigma,\sigma')$.

\section{Further Generalizations}
\label{Markov:generalizations}
The two previous Sections cover the cases of time-homogeneous Markov chains and
pure jump Markov processes. The key observable for the parameter sensitivity evaluation
is the Relative Entropy Rate which is the time average of the path space relative entropy
as time goes to infinity:
\begin{equation}
\RELENTR{{\PATHS}}{{\PATHSAPP}} = \lim_{T\rightarrow\infty} \frac{1}{T} \RELENT{{\PATHS}}{{\PATHSAPP}} \, .
\end{equation}
Additionally, RER has an explicit formula in both cases making it computationally
tractable as we practically demonstrate in Section~\ref{num:sec}.
Thus, if there are more general stochastic processes which also have
an explicit formula for the RER, Fisher Information Matrix can be defined analogously
and gradient-free sensitivity analysis is also doable. Next, we present two families
of stochastic processes which have known RER.

\medskip
\noindent
{\em Time-periodic Markov Processes\,}: Such Markov processes are typically
utilized to describe circular physical or biological phenomena such as annual climate
models or daily behavior of mammals. 
Even though more general classes of processes can be presented, we restrict to
the discrete-time Markov chains with finite state space $E$. The time-inhomogeneous
transition probability matrix is denoted by $p(\sigma,\sigma';m)$ and the periodicity
implies that $p(\sigma,\sigma';m)=p(\sigma,\sigma';k\zeta+m),\ \forall k\in\mathbb Z^+$
where $\zeta$ is the period. Assume that for all $m=0,...,\zeta-1$ the process
$\{\sigma_{k\zeta+m}\}_{k=0}^\infty$ which is a Markov chain has a unique
stationary distribution  $\mu(x,m)$. Then the Markov process $\big\{\sigma_m\big\}_{m\in\mathbb Z^+}$
at steady state regime is periodically stationary with periodic stationary distribution
$\mu$.

In terms of sensitivity analysis, the relative entropy rate between the path
probabilities  has the explicit formula {\small
\begin{equation}
\begin{aligned}
&\RELENTR{{Q_{0,M}^{\theta}}}{{Q_{0,M}^{\theta+\epsilon}}}
= \frac{1}{\zeta} \sum_{m=0}^{\zeta-1} \sum_{\sigma,\sigma'\in E} \mu^{\theta}(\sigma,\zeta) \\
&\ \ \ \ \ \ \ \ \ \ \ \ \times p^{\theta}(\sigma,\sigma';m) \log\frac{p^{\theta}(\sigma,\sigma';m)}{p^{\theta+\epsilon}(\sigma,\sigma';m)} \\
&= \frac{1}{\zeta} \mathbb E_{\mu^\theta} \left[\sum_{m=0}^{\zeta-1} \sum_{\sigma'\in E} 
p^{\theta}(\sigma,\sigma';m) \log\frac{p^{\theta}(\sigma,\sigma';m)}{p^{\theta+\epsilon}(\sigma,\sigma';m)} \right]\, .
\end{aligned}
\end{equation}
} Similar to the previous cases, a generalized formula for the path-space FIM can be
derived. It is given by {\small
\begin{equation}
\begin{aligned}
\FISHERR({Q_{0,M}^{\theta}}) &:= \frac{1}{\zeta} \sum_{m=0}^{\zeta-1} \sum_{\sigma,\sigma'\in E}
\mu^{\theta}(\sigma,\zeta) p^{\theta}(\sigma,\sigma';m) \\
&\times \nabla_\theta \log p^\theta(\sigma,\sigma';m) \nabla_\theta \log p^\theta(\sigma,\sigma';m)^T \, .
\end{aligned}
\end{equation}
} Existence of the relative entropy rate for general time-inhomogeneous Markov chains
can also be found \cite{Wen:96}.

\medskip
\noindent
{\em Semi-Markov Processes\,}:
These processes generalize the jump Markov processes as well as the
renewal processes to the case where
the future evolution (i.e., waiting times and transition probabilities) depends on
the present state and on the time elapsed since the last transition. Semi-Markov
processes have been extensively used to describe reliability models \cite{Limnios:01},
modeling earthquakes \cite{Lutz:93}, queuing theory \cite{Janssen:06}, etc.
In order to define a semi-Markov process the definition of a semi-Markov transition
kernel as well as its corresponding renewal process is required. Let $E$ be a countable state space
then the process $\{J_n,S_n\}_{n\in\mathbb Z^+}$ is a renewal Markov process
with semi-Markov transition kernel $q(\sigma,\sigma';t)\ \sigma,\sigma'\in E,\ t\in\mathbb R^+$ if {\small
\begin{equation}
\begin{aligned}
&\mathbb P\{J_{n+1}=\sigma',S_{n+1}-S_n<t | J_n=\sigma,...,J_0,S_n,...,S_0) \} \\
&=\mathbb P\{(J_{n+1}=\sigma',S_{n+1}-S_n<t | J_n = \sigma \} := q(\sigma, \sigma';t) \, .
\end{aligned}
\end{equation}
} The process $J_n$ is a Markov chain with transition probability matrix elements
$p(\sigma,\sigma')=\lim_{t\rightarrow\infty} q(\sigma,\sigma',t)$ while the process
$S_n$ is the sequence of jump times.
Let $N_t,\ t\in\mathbb R^+$ defined by $N_t = \sup\{n\geq0 : S_n<t$ be the
counting process of the jumps in the interval $(0,t]$. Then the stochastic process
$Z_t,\ t\in\mathbb R^+$ defined by $Z_t = J_{N_t}$ for $t\geq0$ (or $J_n = Z(S_n)$
for $n\geq0$) is the semi-Markov process associated with $(J_n,S_n)$.

Assume further that the (embedded) Markov chain $J_n$ has a stationary distribution
denoted by $\mu$ as well that the mean sojourn time with respect to the stationary
distribution defined by $\hat{m} := \sum_{\sigma,\sigma'\in E}\mu(\sigma) \int_0^\infty q(\sigma,\sigma';t)$
is finite. Then it was shown in \cite{Girardin:03} that the relative entropy rate of
the semi-Markov process $Z_t$ with model parameter vector $\theta$ w.r.t. the
semi-Markov process $\tilde{Z}_t$ with parameter vector $\theta+\epsilon$ is
given by {\small
\begin{equation}
\begin{aligned}
&\RELENTR{{\PATHS}}{{\PATHSAPP}} = \\
&\frac{1}{\hat{m}} \int_0^\infty \sum_{\sigma,\sigma'\in E} \mu^{\theta}(\sigma) q^{\theta}(\sigma,\sigma';s)
\log\frac{q^\theta(\sigma,\sigma';s)}{q^{\theta+\epsilon}(\sigma,\sigma';s)} ds \, ,
\end{aligned}
\end{equation}
}while the Fisher information matrix is similarly defined as {\small
\begin{equation}
\begin{aligned}
\FISHERR{{\PATHS}} &:= \frac{1}{\hat{m}} \int_0^\infty \sum_{\sigma,\sigma'\in E} \mu^\theta(\sigma) q^\theta(\sigma,\sigma';s) \\
&\times \nabla_\theta \log q^\theta(\sigma,\sigma';s) \nabla_\theta \log q^\theta(\sigma,\sigma';s)^T\,ds \, .
\end{aligned}
\end{equation}
}

\section{Numerical Examples}
\label{num:sec}
We demonstrate the wide applicability of the proposed methods by studying the parameter
sensitivity analysis of three models  with very different features and range of applicability. 
Namely, we discuss the  Schl\"ogl model, reversible and irreversible Langevin processes
and the spatially extended ZGB model. Each of these models reveals different aspects of
the proposed method. However, we will first need to discuss  the necessary statistical
estimators  for the Relative Entropy Rate and the Fisher Information Matrix.

\subsection{Statistical Estimators for RER and FIM}
\label{num:approx}
The Relative Entropy Rate \VIZ{RER:MC}, \VIZ{RER:MP} as well as the Fisher Information Matrix \VIZ{FIM:MC},
\VIZ{FIM:MP}
are observables of the
stochastic process and can be estimated as ergodic averages. Thus, both observables are computationally tractable since
they also depend only on the local transition quantities. We discuss each case separately next.

\medskip
\noindent
{\em Discrete-time Markov Chains\,}:  A statistical  estimator for Markov Chains is  directly obtained from 
\VIZ{RER:MC}.  For instance, in the continuous state
space case, the $n$-sample numerical RER is given by {\small
\begin{equation}
\bar{\mathcal H}_1^{(n)} =  \frac{1}{n} \sum_{i=0}^{n-1}
\int_E p^\theta(\sigma_i,\sigma') \log \frac{p^\theta(\sigma_i,\sigma')}{p^{\theta+\epsilon}(\sigma_i,\sigma')} d\sigma'
\label{RER:num:approx:MC1}
\end{equation}}
while the $n$-sample statistical  estimator for FIM is  {\small
\begin{equation}
\bar{\bf F}_1^{(n)} = \frac{1}{n} \sum_{i=0}^{n-1} \int_E p^\theta(\sigma_i,\sigma')
\nabla_\theta \log p^\theta(\sigma_i,\sigma') \nabla_\theta \log p^\theta(\sigma_i,\sigma')^T d\sigma' ,
\label{FIM:num:approx:MC1}
\end{equation}}
where $\{\sigma_i\}_{i=0}^n$ is a realization of the Markov chain with parameter vector $\theta$ at steady (stationary) state.
Thus the  RER for various different perturbation directions (i.e., different $\epsilon$'s)
is  computed from a single run since only the unperturbed process is needed to be simulated.
However, the integrals in \VIZ{RER:num:approx:MC1} and \VIZ{FIM:num:approx:MC1} are rarely
explicitly computable thus a second statistical  estimator for both RER and FIM is obtained  from
the Radon-Nikodym derivative \VIZ{RNderiv:chain} in the path space. It is given by {\small
\begin{equation}
\bar{\mathcal H}_2^{(n)} = \frac{1}{n} \sum_{i=0}^{n-1}
\log \frac{p^\theta(\sigma_i,\sigma_{i+1})}{p^{\theta+\epsilon}(\sigma_i,\sigma_{i+1})}
\label{RER:num:approx:MC2}
\end{equation}}
 while the second estimator  for FIM is {\small
\begin{equation}
\bar{\bf F}_2^{(n)} = \frac{1}{n} \sum_{i=0}^{n-1}
\nabla_\theta \log p^\theta(\sigma_i,\sigma_{i+1}) \nabla_\theta \log p^\theta(\sigma_i,\sigma_{i+1})^T \ .
\label{FIM:num:approx:MC2}
\end{equation}}
Even though, the second  approach is tractable for any transition probability
function,  it suffers from larger variance (see also Fig.~\ref{RER:inTime:fig}), since 
 the summation over all the possible states in \VIZ{RER:num:approx:MC1}
results in estimators with less variance compared to the variance of estimator
\VIZ{RER:num:approx:MC2}. Hence, the first numerical estimator is preferred
whenever applicable (for instance, when the state space is finite and relatively small).
Finally, the estimators  are valid also for time inhomogeneous Markov chain
where $p^\theta(\sigma_i,\sigma_{i+1})$ is replaced by $p^\theta(\sigma_i,\sigma_{i+1};i)$.

\medskip
\noindent
{\em Continuous-time  Markov Chains\,}: The estimators for CTMC are constructed
along the same lines. Indeed, the first estimator for RER is based on \VIZ{RER:MP}
and it is given by {\small
\begin{equation}
\begin{aligned}
&\bar{\mathcal H}_1^{(n)} =  \frac{1}{T} \sum_{i=0}^{n-1} \Delta\tau_i \Big[ \sum_{\sigma'\in E} c^\theta(\sigma_i,\sigma')\\
&\times \log \frac{c^\theta(\sigma_i,\sigma')}{c^{\theta+\epsilon}(\sigma_i,\sigma')}
- \big(\lambda^\theta(\sigma_i) - \lambda^{\theta+\epsilon}(\sigma_i)\big) \Big]
\label{RER:num:approx:MP}
\end{aligned}
\end{equation}}
where $\Delta\tau_i$ is an exponential random variable with parameter $\lambda(\sigma_i)$
while $T=\sum_i \Delta\tau_i$ is the total simulation time. The sequence $\{\sigma_i\}_{i=0}^n$
is the embedded Markov chain with transition probabilities $p^\theta(\sigma_i,\sigma') =
\frac{c^\theta(\sigma_i,\sigma')}{\lambda(\sigma_i)}$ at step $i$. Notice that the weight
$\Delta\tau_i$ at each step which is the waiting time at state $\sigma_i$ is necessary for the
correct estimation of the observable \cite{Gillespie:76}. Similarly, the  estimator for the
FIM is {\small
\begin{equation}
\bar{\bf F}_1^{(n)} = \frac{1}{T} \sum_{i=0}^{n-1} \Delta\tau_i \sum_{\sigma'\in E} 
c^\theta(\sigma_i,\sigma') \nabla_\theta \log c^\theta(\sigma_i,\sigma') \nabla_\theta \log c^\theta(\sigma_i,\sigma')^T \ .
\label{FIM:num:approx:MP}
\end{equation}}
Notice that the computation of the local transition rates $c^\theta(\sigma_i,\sigma')$ for
all $\sigma'\in E$ is needed for the simulation of the jump Markov process when Monte Carlo
methods such as stochastic simulation algorithm (SSA) \cite{Gillespie:76} is utilized.
Thus, the computation of the perturbed transition rates is the only additional computational
cost of this numerical approximation. On the other hand, the second numerical estimator
for RER is based on the Girsanov representation of the Radon-Nikodym derivative
(i.e., \VIZ{RN:MP}) and it is given by {\small
\begin{equation}
\bar{\mathcal H}_2^{(n)} = \frac{1}{n} \sum_{i=0}^{n-1}
 \log \frac{c^\theta(\sigma_i,\sigma_{i+1})}{c^{\theta+\epsilon}(\sigma_i,\sigma_{i+1})}
- \frac{1}{T} \sum_{i=0}^{n-1} \Delta\tau_i \big(\lambda^\theta(\sigma_i) - \lambda^{\theta+\epsilon}(\sigma_i)\big)
\label{RER:num:approx:MP2}
\end{equation}}
Similarly we can construct an FIM estimator.
Notice that the term  in \VIZ{RER:num:approx:MP2} involving logarithms should not be weighted since the
counting measure is approximated with this estimator. 
Unfortunately, the estimator \VIZ{RER:num:approx:MP2} has the same computational cost as
\VIZ{RER:num:approx:MP} due to the need for the computation of the total rate which
is the sum of the local transition rates. Furthermore, in terms of variance, the latter
estimator has worse performance due to the discarded sum over the states $\sigma'$. 

Finally, we complete this section with a proposition that states that all the proposed
estimators are unbiased. 
\begin{proposition}
Under the assumptions of Proposition \ref{propos:MC} for Markov chains or of
Proposition \ref{propos:MP} for jump Markov processes, the numerical estimators
\VIZ{RER:num:approx:MC1}--\VIZ{RER:num:approx:MP2} are unbiased.
\end{proposition}
\begin{proof}
The proofs for each estimator are similar and they are more or less hidden
in the proofs of Propositions \ref{propos:MC} and \ref{propos:MP}. Nevertheless,
we provide the proof for the estimator \VIZ{RER:num:approx:MC2} for the sake
of completeness. We have that
\begin{equation*}
\begin{aligned}
&\mathbb E_Q \big[\bar{\mathcal H}_2^{(n)} \big] = \int\cdot\cdot\cdot \int \frac{1}{n} \sum_{i=0}^{n-1}
\log \frac{p^\theta(\sigma_i,\sigma_{i+1})}{p^{\theta+\epsilon}(\sigma_i,\sigma_{i+1})} \\
&\times \mu^\theta(\sigma_0)
p^\theta(\sigma_0,\sigma_1) \cdot\cdot\cdot p^\theta(\sigma_{n-1},\sigma_n) d\sigma_0\cdot\cdot\cdot d\sigma_n \\
&= \frac{1}{n} \sum_{i=0}^{n-1} \int\int \log \frac{p^\theta(\sigma_i,\sigma_{i+1})}{p^{\theta+\epsilon}(\sigma_i,\sigma_{i+1})}
\mu^\theta(\sigma_i)p^\theta(\sigma_i,\sigma_{i+1}) d\sigma_i d\sigma_{i+1} \\
&= \RELENTR{Q^{\theta}}{Q^{\theta+\epsilon}}
\end{aligned}
\end{equation*}
which completes the proof.
\end{proof}

\subsection{Schl\"ogl Model}
\label{num:sec:Schlogl}
The Schl\"ogl model describes a well-mixed chemical reaction network between three
species $A,\ B,\ X$ \cite{Schlogl:72, Vellela:09}. The concentrations $A,\ B$ are kept
constant while the reaction
rates $k_1,...,k_4$ are the parameters of the model. Table~\ref{Schlogl:rates:table}
provides the propensity functions (rates) for these reactions where $\Omega$ is the
volume of the system Note that $\Omega$ serves as a normalization for the reaction
rates making them of the same order.  Thus, there is no need to resort in logarithmic
sensitivity analysis even though this is possible (see Appendix~\ref{log:sens:analysis:app}). 

\begin{table}[!htb]
\begin{center}
\caption{The rate of the $k$th event when the number of $X$ molecules is $x$
is denoted by $c_k(x)$.  $\Omega$ is the volume of the system.}
\begin{tabular}{|c|c|c|} \hline
Event & Reaction & Rate \\ \hline \hline
1 & $A+2X \rightarrow 3X$ & $c_1(x) = k_1A x(x-1)/(2\Omega)$ \\ \hline
2 & $3X \rightarrow A+2X$ & $c_2(x) = k_2 x(x-1)(x-2) / (6\Omega^2)$ \\ \hline
3 & $B \rightarrow X$ & $c_3(x) = k_3B\Omega$ \\ \hline
4 & $X \rightarrow B$ & $c_4(x) = k_4 x$ \\ \hline
\end{tabular}
\label{Schlogl:rates:table}
\end{center}
\end{table}

The stochastic process describing the number of $X$ molecules of the Schl\"ogl model
is a CTMC with rates provided in Table~\ref{Schlogl:rates:table}.
Since the Schl\"ogl model is a birth/death process, the exact stationary distribution $\mu(x)$,
can be iteratively computed from the reaction rates utilizing the detailed balance condition
\cite{Gardiner:85}. It states that
\begin{equation}
\mu(x) c(x,x+1) = \mu(x+1) c(x+1,x)
\end{equation}
where $c(x,x+1) = c_1(x)+c_3(x)$ is the birth rate at state $x$ while $c(x,x-1)
= c_2(x)+c_4(x)$ is the death rate of the same state. Having the exact stationary
distribution a simple benchmark for the sensitivity of the system is provided.
Furthermore for the parameter values in
Table~\ref{Schlogl:param:val}, the stationary distribution of the Schl\"ogl model
possesses two most probable constant steady states (see also Fig.~\ref{perturbed:inv:meas:fig},
solid lines). Thus, the stochastic process is  non-Gaussian  and Gaussian
approximations \cite{Komorowski:11} are invalid, at least at long times where
transitions between the most likely states take place, see (see Figs.~\ref{RER:inTime:fig}
and \ref{perturbed:inv:meas:fig}). Capturing these transitions is a crucial element for
the correct calculation of stationary dynamics and the efficient sampling of the
stationary distribution. Notice also that there
are studies on sensitivity analysis \cite{Gunawan:05, Degasperi:08} where
the Schl\"ogl model with volume $\Omega=100$ has been used for benchmarking,
however, for this value of $\Omega$  the most likely states in Fig.~\ref{perturbed:inv:meas:fig}
are steep and the simulation algorithm is trapped, depending
on the initial data, into the one of the two corresponding wells. Thus, for deep
wells it takes an  exponentially long time to make a transition from one to the
other well, consequently, the sensitivity analysis is biased and depends on the
initial value of the process. In fact, in the case of deep wells the Gaussian
approximation is correct and the FIM analysis \cite{Komorowski:11} applies
as long as the process remains trapped. In a intuitive sense, the volume $\Omega$
can be thought as the inverse temperature of the system making the stationary
distribution more or less steep \cite{Hanggi:84}.

\begin{table}[!htb]
\centering
\caption{Parameter values for the Schl\"ogl model.}
\begin{tabular}{|c||c|c|c|c|c|} \hline
Parameters & $\Omega$ & $k_1A$ & $k_2$ & $k_3B$ & $k_4$ \\ \hline
Values & 15 & 3 & 1 & 2 & 3.5 \\ \hline
\end{tabular}
\label{Schlogl:param:val}
\end{table}

\begin{figure}[!htb]
\begin{center}
\includegraphics[width=0.5\textwidth]{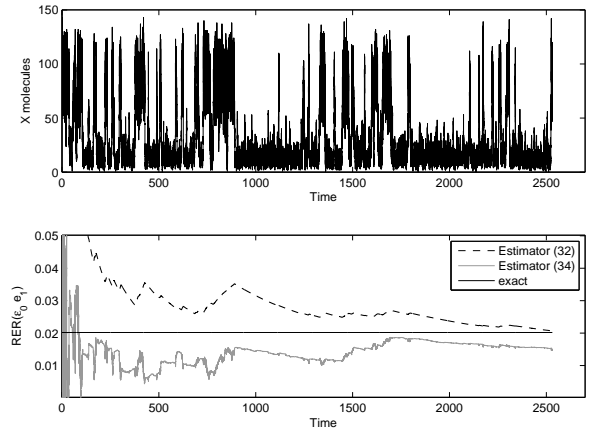}
\caption{Upper plot: The number of $X$ molecules as a function of time.
The stochastic process sequentially visits the two most probable states
defined as the maxima of the PDF. Lower
panel: RER as a function of time when $k_1A$ is perturbed by 0.05
computed using \VIZ{RER:num:approx:MP} (dashed line) and using
\VIZ{RER:num:approx:MP2} (grey line). In both cases, the accuracy of
the numerical estimators increase as the number of samples increases.}
\label{RER:inTime:fig}
\end{center}
\end{figure}

Let denote $\theta=[k_1A,k_2,k_3B,k_4]^T$, then the numerical estimator for
RER as well as for FIM for the Schl\"ogl model is given by \VIZ{RER:num:approx:MP}
and \VIZ{FIM:num:approx:MP}, respectively.
The upper panel of Fig.~\ref{RER:inTime:fig} shows the number of $X$
molecules in the course of time. The number of jumps of this process are
$10^6$ while the initial value $X_0=100$ is slightly above the minimum
of the second well. The lower panel of Fig.~\ref{RER:inTime:fig}
shows the numerical RER (dashed line) as a function of time when only $k_1A$ is
perturbed by $0.05$ (i.e., perturbation is $\epsilon=0.05 e_1$) as well as the exact
RER computed from (\ref{RER:MP}). For comparison purposes, we also plot the
RER estimator \VIZ{RER:num:approx:MP2}. Obviously, as simulation time is increased
both numerical RER estimators converge to the exact value even though the
estimator \VIZ{RER:num:approx:MP2} needs more samples to converge (i.e.,
its variance is larger). Notice that enough transitions
between the two steady states are necessary in order to obtain robust results.
Fig.~\ref{RER:directions:fig} depicts the exact RER (circles), the
numerically-computed RER (stars) as well the FIM-based RER (squares).
The directions $\pm\epsilon_0 e_k, \ k=1,...,4$ where $\epsilon_0$ is set
to $0.05$ while $e_k$ are the typical orthonormal unit vectors are
considered. These directions correspond to the perturbation of just one of the
model's parameters. The number of jumps of this simulation is $5\cdot10^6$
while the initial value is again $X_0=100$. The numerically-computed RERs
have similar values with the exact ones as Fig.~\ref{RER:directions:fig}
demonstrates. The computed RERs imply that the most sensitive parameter
is $k_2$ (corresponds to $\pm e_2$) while the least sensitive parameter
is $k_3B$ (corresponds to $\pm e_3$). Another important feature of the
proposed sensitivity method is that the
RERs for all the different parameter perturbations are computed from a single
simulation run of the unperturbed process. Thus, for each direction, the only additional
computational cost is the calculation of the perturbed rates of the process.
Notice also that RER gives different values between a direction and its opposite
resulting in assigning different sensitivities while FIM-based RER cannot distinguish
between the two opposite directions since it is a second-order (quadratic)
approximation.

\begin{figure}[!htb]
\begin{center}
\includegraphics[width=0.5\textwidth]{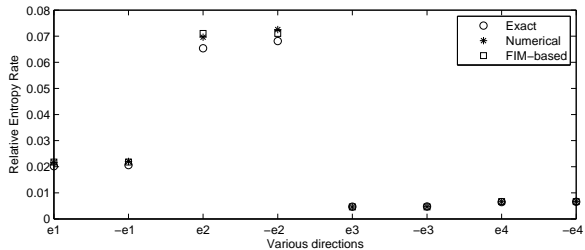}
\caption{Exact (circles), numerical (stars) and FIM-based (squares) RER for various
directions. $k_2$ is the most sensitive parameter followed by $k_1A$ while
the least sensitive parameters are $k_4$ and $k_3B$.}
\label{RER:directions:fig}
\end{center}
\end{figure}

We further validate the inference capabilities of RER by illustrating the actual stationary distribution
 of the perturbed processes. It is expected that the most/least sensitive
parameters of the path distribution should be strongly related with the most/least
sensitive parameters of the stationary distribution. Indeed, the upper panel of
Fig.~\ref{perturbed:inv:meas:fig} presents the stationary distributions of the
unperturbed process (solid line) as well the perturbed stationary distribution of the
most (dashed line) and least (dotted line) sensitive parameters. The perturbation
of the most sensitive parameter results in the largest change of the 
stationary distribution while the smallest change is observed when the least sensitive parameter
is perturbed. Moreover, FIM can be used for the computation not only of the most
sensitive parameter but also for the computation of the most sensitive direction
in general. Indeed, the most sensitive direction can be found by performing
eigenvalue analysis to the FIM. The eigenvector with the highest eigenvalue
defines the most sensitive direction. In our setup, the most sensitive direction
is $\epsilon_{\max} = [0, 0.978, 0, 0.207]$. The prominent parameter of the
most sensitive direction is $k_2$ which is not a surprise since, from
Fig.~\ref{RER:directions:fig}, $k_2$ is the most sensitive parameter. The
lower panel of Fig.~\ref{perturbed:inv:meas:fig} depicts the 
stationary distribution of the most sensitive parameter (i.e., $k_2$ or $-\epsilon_0 e_2$)
(dashed line) and the most sensitive direction (i.e., $\epsilon_0\epsilon_{\max}$)
(dotted line). It is evident that the stationary distribution of the most sensitive
direction is further away from the unperturbed stationary distribution compared
to the stationary distribution of the most sensitive parameter.

\begin{figure}[!htb]
\begin{center}
\includegraphics[width=0.5\textwidth]{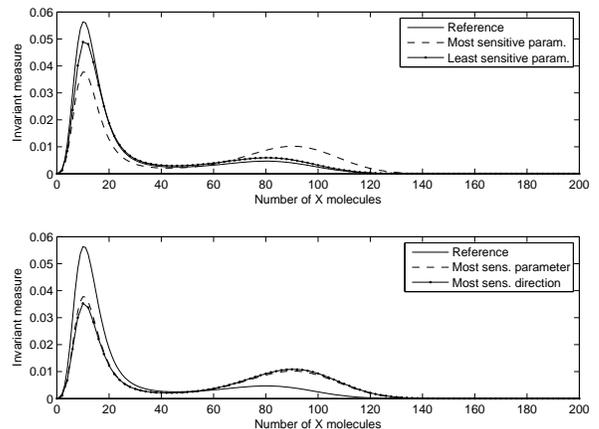}
\caption{Upper plot: The stationary distributions for the unperturbed process (solid line), the most sensitive parameter $k_2$
(dashed line) and the least sensitive paramter $k_3B$ (dotted line).
Lower plot: The stationary distributions for the unperturbed process (solid line), the most sensitive parameter $k_2$
(dashed line) and the most sensitive direction $\epsilon_{\max}$ (dotted line).}
\label{perturbed:inv:meas:fig}
\end{center}
\end{figure}

\subsection{Reversible and non-reversible Langevin Processes}
\label{num:sec:Langevin}
The second example we consider  is a particle model with interactions which have been applied
and studied primarily in molecular dynamics \cite{Rapaport:95, Schlick:02, Frenkel:02,
Lelievre:10} but also in biology (for instance, in swarming \cite{Ebeling:08}),
etc. In molecular dynamics, the Langevin dynamics is typically a
Hamiltonian system coupled with a thermostat (i.e., noise).
A Langevin process is defined by the SDE system {\small
\begin{equation}
\begin{aligned}
&dq_t = \frac{1}{m}p_t dt \\
&dp_t = - {\bf F}(q_t)dt - \frac{\gamma}{m} p_t dt + \sigma dB_t \\
\label{Langevin:sde:cont}
\end{aligned}
\end{equation}}
where $q_t\in\mathbb R^{dN}$ is the position vector of the $N$ particles
in $d$ dimensions,
$p_t\in\mathbb R^{dN}$ is the momentum vector of the particles, $m$ is
the mass of the particles, $\bf F$ is a driving force, $\gamma$ is the friction factor,
$\sigma$ is the diffusion factor and $B_t$ is a $dN$-dimensional Brownian
motion. The first equation which describes the evolution of the position
of the particles is deterministic thus the overall SDE system is degenerate.
In the zero-mass limit or the infinite-friction limit, Langevin process is simplified
to overdamped Langevin process which is non-degenerate, however, several
studies advocate the use of Langevin dynamics directly \cite{Scemama:06, Cances:07}.
The proposed sensitivity analysis approach is widely applicable to
SDE systems once the  assumption on ergodicity is  satisfied.


The vector field  ${\bf F}(\cdot)$ denotes the force exerted on the system and here we assume
it consists of two terms:  a gradient (potential) component as in typical Langevin systems, as well
as an additional non-gradient term, where the latter is assumed to be   divergence-free:
\begin{equation}\label{Force}
{\bf F}(q) = \nabla_q V(q) + \alpha G(q)\, ,
\end{equation}
and  $\nabla_q \cdot G=0$.
Here we consider particular examples to illustrate the applicability of the proposed sensitivity analysis methods.
The gradient term in \VIZ{Force}  models particle interactions
given by
\begin{equation}
V(q) = \sum_{i,j<i} V_M (|q_i-q_j|)
\end{equation}
where $V_M(r)$ is the three-parameter Morse potential $V_M(r) = D_e (1 - e^{-a(r-r_e)})^2$.
The Morse potential includes  a combination of  short-range  repulsive and long-range attractive
interactions and has been extensively used in molecular simulations \cite{Kaplan:03}. The
divergence-free component is taken to be a simple antisymmetric force given by
\begin{equation}
G_i(q) = q_{i+1} - q_{i-1}\ ,\ \ \ i=1,...,N \\
\end{equation}
where $q_0=q_N$ and $q_{N+1}=q_1$.

We now return to \VIZ{Force} and discuss the implications of its structure. When $\alpha=0$, the Langevin
process is reversible meaning that the condition of detailed balance is satisfied with respect to a known
Gibbs stationary probability  distribution \cite{Lelievre:10}. However,  knowing  the stationary distribution
explicitly is insufficient to  carry out sensitivity analysis on the stationary dynamics which typically 
may include dynamic transitions between metastable states, as in the Schl\"ogl Model discussed earlier.
Furthermore,  when $\alpha\neq0$, detailed balance does not hold true in general and the stationary probability
distribution of the corresponding Langevin process is not known since the system is non-reversible
\cite{Lebowitz:99, Maes:00}. Examples of  forces  such as \VIZ{Force} that include non-gradient terms
and yield  non-reversible Langevin equations, arise  typically in driven systems, for instance in Brownian
particle suspensions where particles interact with a fluid flow \cite{Alber:11}. For non-reversible
systems  no efficient method for sensitivity analysis has been reported in the literature,  at least for
the cases dealt here, namely (a) long-time, stationary dynamics  (also referred to as non-equilibrium
steady states (NESS) \cite{Lebowitz:99, Maes:00}), as well as, (b)  the unknown stationary probability.
Our proposed path-space RER sensitivity methods can address these challenges and is straightforwardly
applicable to both reversible and non-reversible Langevin equations as we show next.

First, an explicit EM--Verlet (symplectic)--implicit EM scheme is applied for
the discretization of (\ref{Langevin:sde:cont}). It is written as {\small
\begin{equation}
\begin{aligned}
p_{i+\frac{1}{2}} &= p_i -{\bf F}(q_i)\frac{\Delta t}{2} - \frac{\gamma}{m}p_i \frac{\Delta t}{2} + \sigma \Delta W_i \\
q_{i+1} &= q_i + m^{-1}p_{i+\frac{1}{2}} \Delta t \\
p_{i+1} &= p_{i+\frac{1}{2}} -{\bf F}(q_{i+1})\frac{\Delta t}{2} - \frac{\gamma}{m}p_{i+1} \frac{\Delta t}{2} + \sigma \Delta W_{i+\frac{1}{2}} \\
\label{Langevin:BBK:integrator}
\end{aligned}
\end{equation}}
with $\Delta W_i, \Delta W_{i+\frac{1}{2}} \sim {\it N}(0, \frac{\Delta t}{2} I_{dN})$ where
${\it N}$ is the multivariate normal distribution. This
numerical scheme also known as BBK integrator \cite{Brunger:84, Lelievre:10} utilizes
a Strang splitting. Thus, the discretized Langevin process is a Markov chain with continuous
state space. Notice that the numerical scheme is non-degenerate, thus, the transition
probability from state $(q,p)$ to state $(p',q')$ is given by
\begin{equation}
P(q,p,q',p') = P(q'|q,p) P(p'|q',q,p)
\end{equation}
where
\begin{equation}
P(q'|q,p) = \frac{1}{Z_0} e^{-\frac{m^2}{\sigma^2\Delta t^3}
\left|q'-q + (p-{\bf F}(q)\frac{\Delta t}{2m} + p\frac{\Delta t \gamma}{2m})\Delta t\right|^2}
\end{equation}
and
\begin{equation}
P(p'|q',q,p) = \frac{1}{Z_1} e^{-\frac{1}{\sigma^2\Delta t}
\left|(1+\frac{\gamma\Delta t}{2m})p'-(\frac{m}{\Delta t}(q'-q) - \frac{\Delta t}{2}{\bf F}(q')) \right|^2}
\end{equation}
where $Z_i,\ i=0,1$ are the respective normalization constants. Let now define the
parameter vector $\theta = [D_e, a, r_e]$. Then, the discretized Langevin model
(\ref{Langevin:BBK:integrator}) is a discrete-time Markov process with
$\mathbb R^{2dN}$ being the state space. The statistical  estimators for RER
as well as for FIM are given by \VIZ{RER:num:approx:MC2} and
\VIZ{FIM:num:approx:MC2}, respectively. Notice that the estimators with larger
variance were chosen because the integration of the transition probability
density function w.r.t. the positions is not a trivial problem, if not intractable
in high dimensions.

\begin{table}[!htb]
\centering
\caption{Parameter values for the discretized Langevin system.}
\begin{tabular}{|c||c|c|c|c|c|c|c|c|c|} \hline
Parameters & $N$ & $D_e$ & $a$ & $r_e$ & $m$ & $\gamma$ & $\sigma$ & $\Delta t$ \\ \hline
Values & 3 & 0.3 & 0.3 & 1 & 1 & 1 & 0.1 & 0.01 \\ \hline
\end{tabular}
\label{Langevin:param:val}
\end{table}

The upper panel of Fig.~\ref{Langevin:RER} depicts the numerical RER
as a function of simulation time for the parameter values given in Table~\ref{Langevin:param:val}.
The reversible case is considered while the sensitivity of the parameters is obtained
from the directions defined by the orthonormal unit vectors multiplied with $\epsilon_0=0.05$.
Since the initial positions and momenta where randomly chosen from a uniform
distribution an initial out-of-equilibrium time regime can be seen in the Figure
(up to time $t_0=100$). Moreover, the variance of RER
as an observable is rather large which can be explained by the small number
of particles. Systems with more particles are expected to converge faster due to
averaging effects. The lower panel of Fig.~\ref{Langevin:RER} depicts the
RER at final time $t=10^4$ with an initial equilibration time $t_0=100$ where the
numerical RER is discarded. Evidently, the most sensitive parameter is $a$
followed by $D_e$ while the least sensitive parameter is $r_e$.

\begin{figure}[!htb]
\begin{center}
\includegraphics[width=0.5\textwidth]{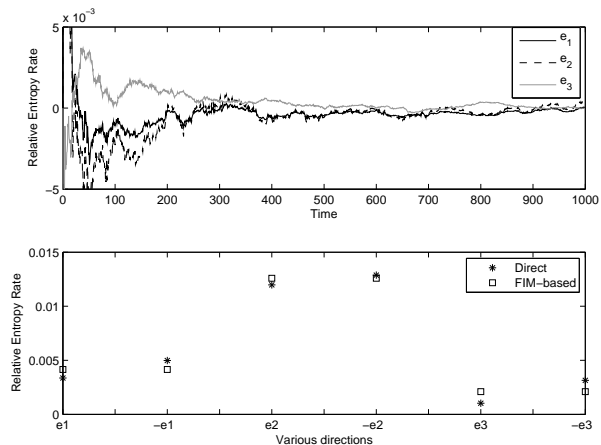}
\caption{Upper plot: Relative entropy rate as a function of time for perturbations of
$D_e$ (solid line), $a$ (dashed line) and of $r_e$ (grey line) at the reversible
regime ($\alpha=0$). The variance of the numerical RER is large, necessitating
more samples for accurate estimation. Lower plot: RER for various directions.
The most sensitive parameter is $a$.}
\label{Langevin:RER}
\end{center}
\end{figure}

Utilizing our methodology the parameter sensitivity of not only the reversible regime
but also of the {\it non-reversible}, $\alpha \ne 0$,  regime can be explored even though the stationary probability is not
known. Fig.~\ref{Langevin:level:sets} shows the level sets of the FIM matrix for
the reversible case (upper plots, $\alpha=0$) and for the irreversible case (lower plots,
$\alpha=0.1$). Figure suggests that the additional irreversible component results in
the fact that some directions became more sensitive and some other directions
became less sensitive. Further validation is obtained from the eigenvalues of the
FIM which are $7.30, 0.592, 0.015$ for the reversible case while the eigenvalues
for the irreversible case are $13.90, 0.302, 0.074$.
Finally, FIM can be very useful in various ways for the quantification of sensitivity
analysis. For instance, the determinant of FIM which in optimal experiment design is called
A-optimality can be used as a measure of parameter identification \cite{Rothenberg:71,
Emery:98, Komorowski:11}. In our particular example, the determinants are $0.065$ and
$0.313$ for the reversible and irreversible cases, respectively. This result asserts
that in the non-reversible case $\alpha \ne 0$ in \VIZ{Force}, the divergence-free
component improves the ability of any estimator of the potential's parameters.

\begin{figure}[!htb]
\begin{center}
\includegraphics[width=0.5\textwidth]{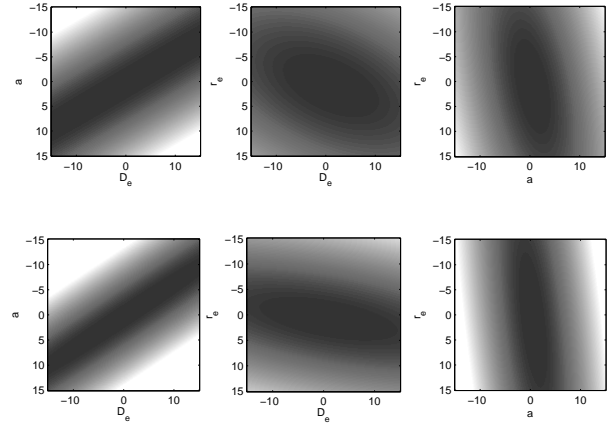}
\caption{Upper plots: Level sets (or neutral spaces) for the reversible case ($\alpha=0$).
Lower plots: Level sets for the irreversible case ($\alpha=0.1$).}
\label{Langevin:level:sets}
\end{center}
\end{figure}

\subsection{Spatially extended Kinetic Monte Carlo models}
\label{num:sec:ZGB}
The applicability of the proposed sensitivity method is further demonstrated in
spatially extended systems which exhibit complex spatio-temporal morphologies
such as islands, spirals, rings, etc. at mesoscale length scales. Among the various
surface mechanisms such as adsorption, desorption, diffusion, etc. we focus
on $CO$ oxidation which is a prototypical example for molecular-level reaction-diffusion
mechanism between adsorbates on a catalytic surface. A simplified $CO$ oxidization model
without diffusion known as the Ziff-Gulari-Barshad (ZGB) model \cite{Ziff:86} is considered.
Despite being an idealized model, the ZGB model incorporates the basic mechanisms
for the dynamics of adsorbate species during $CO$ oxidation on catalytic surfaces,
namely, single site updates (adsorption/desorption) and multisite reactions (two neighboring
sties being involved). Due to the reactions between species, the ZGB model is non-reversible
and its stationary distribution is unknown. Nevertheless, our sensitivity analysis methodology
is capable of quantify the parameter sensitivities utilizing only the rates of the process
which are provided in Table~\ref{ZGB:rates:table}. The spins of the two dimensional
lattice $\Lambda_N$ with $N$ lattice sites take values
$\sigma(j)=0$ denoting a vacant site $j\in\Lambda_N$, $\sigma(j)=-1$ for a $CO$ molecule
at site $j$ and $\sigma(j)=1$ for an $O$ molecule. Depending on the local configuration
of site $j$ as well as of the nearest neighbors, the events with the respective rates provided
in Table~\ref{ZGB:rates:table} are executed.

{\small
\begin{table}
\begin{center}
\caption{The rate of the $k$th event of the $j$th site given that the current configuration
is $\sigma$ is denoted by $c_k(j;\sigma)$ where n.n. stands for nearest neighbors.}
\begin{tabular}{|c|c|c|} \hline
Event & Reaction & Rate \\ \hline \hline
1 & $\emptyset \rightarrow CO$ & $(1-\sigma(j)^2) k_1$ \\ \hline
2 & $\emptyset \rightarrow O_2$ & $(1-\sigma(j)^2) (1-k_1) \frac{\#\text{vacant n.n.}}{\text{total n.n.}}$ \\ \hline
3 & $CO + O \rightarrow CO_2 + \text{des.}$ & $\frac{1}{2}\sigma(j) (1+\sigma(j)) k_2 \frac{\#O\text{ n.n.}}{\text{total n.n.}}$ \\ \hline
4 & $O + CO \rightarrow CO_2 + \text{des.}$ & $\frac{1}{2}\sigma(j) (\sigma(j)-1) k_2 \frac{\#CO\text{ n.n.}}{\text{total n.n.}}$ \\ \hline
\end{tabular}
\label{ZGB:rates:table}
\end{center}
\end{table}}

The ZGB model is a  high-dimensional CTMC
which here is simulated utilizing the stochastic simulation algorithm \cite{Gillespie:76}.
For each step of the simulation, the rates of the process for all sites of the lattice
are needed. Interestingly, in order to perform our sensitivity analysis to the
system parameters, only the rates are incorporated. Indeed, denoting by
$\theta = [k_1,k_2]$ the parameter vector, then the statistical estimators
of RER as well as of FIM for the ZGB model are given by \VIZ{RER:num:approx:MP}
and \VIZ{FIM:num:approx:MP}, respectively. Nevertheless, we explicitly provide
the numerical RER estimator for convenience: {\small
\begin{equation}
\begin{aligned}
\bar{\mathcal H}_1^{(n)} =
\frac{1}{T} \sum_{i=0}^{n-1} &\Delta\tau_i \Big[\sum_{j\in \Lambda_N}\sum_{k=1}^4 c_k^\theta(j;\sigma_i)
\log \frac{c_k^\theta(j;\sigma_i)}{c_k^{\theta+\epsilon}(j;\sigma_i)} \\
&+ \lambda^{\theta+\epsilon}(\sigma_i) - \lambda^\theta(\sigma_i) \Big]
\label{ZGB:RER:num:approx}
\end{aligned}
\end{equation}}
where $c_k^\theta(j;\sigma)$ is the $k$th event of lattice site $j$ when the
lattice configuration is $\sigma$ while $\lambda^\theta(\sigma) =
\sum_{j\in \Lambda_N}\sum_{k=1}^4 c_k^\theta(j;\sigma)$ is the total
rate of the process at state $\sigma$.

The upper panel of Fig.~\ref{ZGB:RER} depicts the RER as a function of simulation time
when $k_1=0.35$ is perturbed by $\epsilon_0=0.02$ (solid line) and when $k_2=0.85$
is perturbed by the same amount. It is evident that after an initial burning time, RER
converges fast to a limit value implying that the variance of RER as an observable
is small. This can be explained by the fact that at each step of the simulation,
(\ref{ZGB:RER:num:approx}) averages the over the entire lattice in order to compute
the instantaneous RER. The lower panel of Fig.~\ref{ZGB:RER} depicts the RER at final
time $t=100$ with an initial equilibration time $t_0=10$ where the instantaneous
RER is discarded. Obviously, the most sensitive parameter is $k_1$ which is related
with the adsorption mechanism while the least sensitive is $k_2$. 
In order to further validate our findings,
we plot the lattice configuration when either $k_1$ or $k_2$ is perturbed by $\epsilon_0=0.02$.
Fig.~\ref{ZGB:config} depicts the configuration of the unperturbed system
as well as the configurations when one of the two model parameters are perturbed.
Evidently, the configuration when the most sensitive parameter (i.e., $k_1$) is
perturbed is less similar to the unperturbed configuration compared to the
configuration when the least sensitive parameter (i.e., $k_2$) is perturbed.

\begin{figure}[!htb]
\begin{center}
\includegraphics[width=0.5\textwidth]{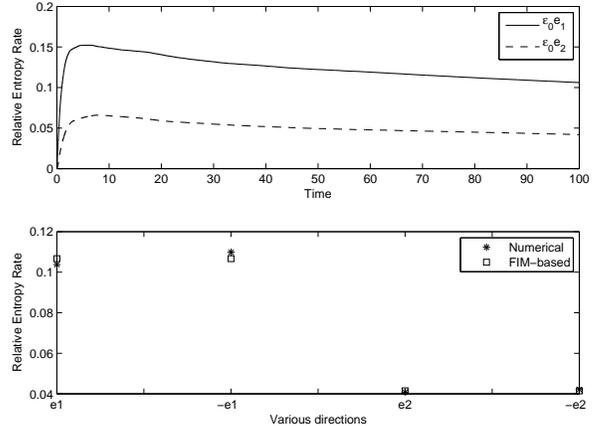}
\caption{Upper plot: Relative entropy rate as a function of time for perturbations of
both $k_1$ (solid line) and of $k_2$ (dashed line). An equilibration time
until the process reach its metastable regime is evident. Lower plot: RER for various
directions. The most sensitive parameter is $k_1$.}
\label{ZGB:RER}
\end{center}
\end{figure}

\begin{figure}[!htb]
\begin{center}
\includegraphics[width=0.5\textwidth]{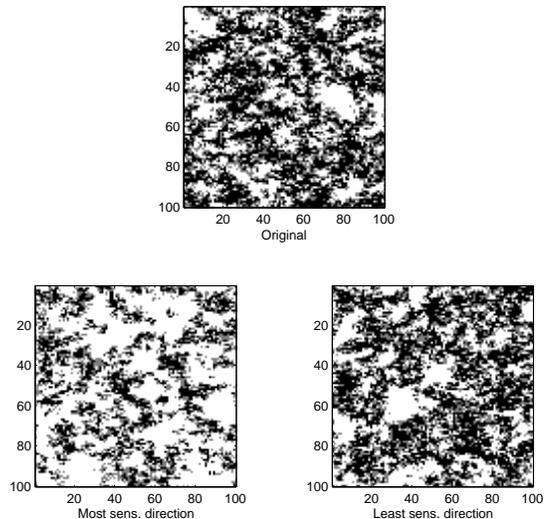}
\caption{Configurations obtained by $\epsilon_0$-perturbations of the most and least
sensitive parameters. The comparison with the reference configuration reveals the
differences between the most and least sensitive perturbation parameters.}
\label{ZGB:config}
\end{center}
\end{figure}

Thus far, we have performed local sensitivity analysis meaning that we were
concentrated around a single point of the parameter space. Even though various
global sensitivity analysis approaches have been derived based on variance \cite{Chan:97,Saltelli:08}
or on mutual information \cite{Ludtke:08}, here, we present a demonstration of global
sensitivity analysis based on a phase diagram of the most and least sensitive
directions. Indeed, any direction can be seen as a vector field and a phase diagram
of a subset of the parameter regime can be visualized. Fig.~\ref{ZGB:RER:phase:diagram}
depicts the most (solid) and least (dashed) sensitive directions which correspond
to the stronger and weaker eigenvalues of the FIM, respectively. Notice that the most/least
sensitive directions are parallel to the axes which asserts that the FIM is diagonal.
This can be explained by the fact that the parameters of the model $k_1$ and
$k_2$ affect different rates in a decoupled fashion (check Table~\ref{ZGB:rates:table}).

Finally, we note  that even though we have considered a spatial KMC model with
few parameters to assess their sensitivity, our emphasis is  primarily on (a) the high
dimensionality of the process, and (b) the non-reversibility of the process without prior
knowledge of  the  stationary probability distribution. For such complex systems
there appears to be no previous systematic work  in the literature on sensitivity
analysis. 

\begin{figure}[!htb]
\begin{center}
\includegraphics[width=0.45\textwidth]{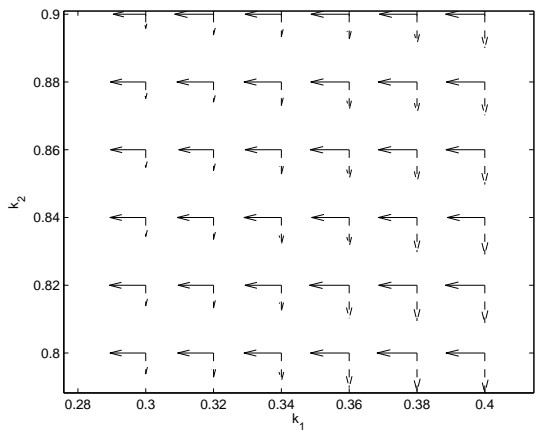}
\caption{Vector field with the most (solid arrows) and least (dashed arrows) sensitive
directions computed from eigenvalue analysis of FIM. The length of the arrows is
proportional to the corresponding eigenvalue.}
\label{ZGB:RER:phase:diagram}
\end{center}
\end{figure}

\section{Conclusions}
\label{concl}
Here we  proposed a novel method for   sensitivity analysis of  complex stochastic dynamics,
based on the concept of Relative Entropy Rate between two stochastic processes. The method
is computationally  feasible  at the stationary  regime and involves the calculation of suitable
observables in path space for the Relative Entropy Rate and the corresponding Fisher Information
Matrix.  The stationary regime is crucial for stochastic dynamics and  can  allow us to address the
sensitivity analysis of  complex systems, including examples of processes with complex landscapes
that exhibit metastability and strong intermittency, non-reversible systems from a statistical mechanics
perspective, and  high-dimensional,  spatially distributed models.  Our proposed methods bypass these
challenges relying  on the direct Monte Carlo simulation of rigorously derived  observables for 
the Relative Entropy Rate and Fisher Information in path space rather than on the  stationary
probability distribution itself. The knowledge of the Fisher Information Matrix provides a gradient-free
method for sensitivity analysis, as well as allows   to address questions of parameter
identifiability and optimal experiment design  in complex stochastic dynamics.

Although the proposed methods are widely applicable to many stochastic models,
we demonstrated their capabilities  by focusing  on  two classes of problems. First,
on Langevin particle systems with either reversible (gradient) or non-reversible
(non-gradient) forcing,  highlighting the ability of the  method to carry out sensitivity analysis
in non-equilibrium systems; second, on spatially extended  Kinetic Monte Carlo  models,
showing that the method can  handle high-dimensional problems. In fact, we showed that
the proposed approach to sensitivity analysis  is suitable for non-equilibrium systems, 
where the structure of the stationary PDF is unknown and is typically non-Gaussian.
Finally, the sensitivity   estimators can be easily embedded in any available  molecular
simulation methods such as Kinetic Monte Carlo or Langevin solvers.

\section*{Acknowledgements}
This work was supported, in part, by the by the Office
of Advanced Scientific Computing Research, U.S. Department of Energy under contract 
DE-SC0002339  and the EU project FP7-REGPOT-2009-1  ``Archimedes Center for Modeling,
Analysis and Computation''. We would like to thank Andrew Majda, Petr Plech{\'a}{\v{c}}, Luc Rey-Bellet
and Dion Vlachos for many interesting and valuable discussions as well as Giorgos Arampatzis
for providing us with the ZGB simulation algorithm.

\appendix

\section{Sensitivity analysis on the logarithmic scale}
\label{log:sens:analysis:app}
In many applications, the model parameters differ by orders of magnitude and
the only meaningful option in order to study sensitivity analysis is to perform
relative parameter perturbations. This is done by perturbing the logarithm
of the model parameters instead of the parameters itself. Thus, utilizing the
chain rule for $\nabla_{\log\theta} f(\theta) = 
\nabla_{\theta} f(\theta) . \nabla_{\log\theta} \theta = \theta . \nabla_{\theta} f(\theta)$
where `$.$' means element by element multiplication, the logarithmic-scale
Fisher information matrix has elements: {\small
\begin{equation}
\big( \FISHERR({Q^{\log\theta}}) \big)_{i,j} = \theta_i \theta_j \big( \FISHERR({Q^{\theta}}) \big)_{i,j}
\ ,\ \ \ \ i,j=1,...,k \ .
\end{equation}}
Similarly, the logarithmic perturbation for the RER is performed by utilizing the
perturbation vector $\theta . \epsilon$ instead of $\epsilon$. Notice that \VIZ{GFIM}
continuous to be valid for the logarithmic scale. Indeed, it holds that {\small
\begin{equation}
\RELENTR{{Q^{\theta}}}{{Q^{\theta(1+\epsilon)}}} = \frac{1}{2}(\theta.\epsilon)^T
\FISHERR({Q^{\log\theta}}) (\theta.\epsilon) + O(|\theta.\epsilon|^3) \ .
\end{equation}}

\small

\begin{thebibliography}{62}%
\makeatletter
\providecommand \@ifxundefined [1]{%
 \@ifx{#1\undefined}
}%
\providecommand \@ifnum [1]{%
 \ifnum #1\expandafter \@firstoftwo
 \else \expandafter \@secondoftwo
 \fi
}%
\providecommand \@ifx [1]{%
 \ifx #1\expandafter \@firstoftwo
 \else \expandafter \@secondoftwo
 \fi
}%
\providecommand \natexlab [1]{#1}%
\providecommand \enquote  [1]{``#1''}%
\providecommand \bibnamefont  [1]{#1}%
\providecommand \bibfnamefont [1]{#1}%
\providecommand \citenamefont [1]{#1}%
\providecommand \href@noop [0]{\@secondoftwo}%
\providecommand \href [0]{\begingroup \@sanitize@url \@href}%
\providecommand \@href[1]{\@@startlink{#1}\@@href}%
\providecommand \@@href[1]{\endgroup#1\@@endlink}%
\providecommand \@sanitize@url [0]{\catcode `\\12\catcode `\$12\catcode
  `\&12\catcode `\#12\catcode `\^12\catcode `\_12\catcode `\%12\relax}%
\providecommand \@@startlink[1]{}%
\providecommand \@@endlink[0]{}%
\providecommand \url  [0]{\begingroup\@sanitize@url \@url }%
\providecommand \@url [1]{\endgroup\@href {#1}{\urlprefix }}%
\providecommand \urlprefix  [0]{URL }%
\providecommand \Eprint [0]{\href }%
\@ifxundefined \urlstyle {%
  \providecommand \doi  [0]{\begingroup \@sanitize@url \@doi}%
  \providecommand \@doi [1]{\endgroup \@@startlink {\doibase
  #1}doi:\discretionary {}{}{}#1\@@endlink }%
}{%
  \providecommand \doi  [0]{doi:\discretionary{}{}{}\begingroup
  \urlstyle{rm}\Url }%
}%
\providecommand \doibase [0]{http://dx.doi.org/}%
\providecommand \Doi [0]{\begingroup \@sanitize@url \@Doi }%
\providecommand \@Doi  [1]{\endgroup\@@startlink{\doibase#1}\@@Doi}%
\providecommand \@@Doi [1]{#1\@@endlink}%
\providecommand \selectlanguage [0]{\@gobble}%
\providecommand \bibinfo  [0]{\@secondoftwo}%
\providecommand \bibfield  [0]{\@secondoftwo}%
\providecommand \translation [1]{[#1]}%
\providecommand \BibitemOpen [0]{}%
\providecommand \bibitemStop [0]{}%
\providecommand \bibitemNoStop [0]{.\EOS\space}%
\providecommand \EOS [0]{\spacefactor3000\relax}%
\providecommand \BibitemShut  [1]{\csname bibitem#1\endcsname}%
\bibitem [{\citenamefont {Gunawan}\ \emph {et~al.}(2005)\citenamefont
  {Gunawan}, \citenamefont {Cao}, \citenamefont {Petzold},\ and\ \citenamefont
  {III}}]{Gunawan:05}%
  \BibitemOpen
  \bibfield  {author} {\bibinfo {author} {\bibfnamefont {R.}~\bibnamefont
  {Gunawan}}, \bibinfo {author} {\bibfnamefont {Y.}~\bibnamefont {Cao}},
  \bibinfo {author} {\bibfnamefont {L.}~\bibnamefont {Petzold}}, \ and\
  \bibinfo {author} {\bibfnamefont {F.~J.~D.}\ \bibnamefont {III}},\ }\bibfield
   {title} {\enquote {\bibinfo {title} {Sensitivity analysis of discrete
  stochastic systems},}\ }\href@noop {} {\bibfield  {journal} {\bibinfo
  {journal} {Biophysical Journal},\ }\textbf {\bibinfo {volume} {88}},\
  \bibinfo {pages} {2530--2540} (\bibinfo {year} {2005})}\BibitemShut {NoStop}%
\bibitem [{\citenamefont {Nakayama}\ \emph {et~al.}(1994)\citenamefont
  {Nakayama}, \citenamefont {Goyal},\ and\ \citenamefont
  {Glynn}}]{Nakayama:94}%
  \BibitemOpen
  \bibfield  {author} {\bibinfo {author} {\bibfnamefont {M.}~\bibnamefont
  {Nakayama}}, \bibinfo {author} {\bibfnamefont {A.}~\bibnamefont {Goyal}}, \
  and\ \bibinfo {author} {\bibfnamefont {P.~W.}\ \bibnamefont {Glynn}},\
  }\bibfield  {title} {\enquote {\bibinfo {title} {Likelihood ratio sensitivity
  analysis for {M}arkovian models of highly dependable systems},}\ }\href@noop
  {} {\bibfield  {journal} {\bibinfo  {journal} {Stochastic Models},\ }\textbf
  {\bibinfo {volume} {10}},\ \bibinfo {pages} {701--717} (\bibinfo {year}
  {1994})}\BibitemShut {NoStop}%
\bibitem [{\citenamefont {Plyasunov}\ and\ \citenamefont
  {Arkin}(2007)}]{Plyasunov:07}%
  \BibitemOpen
  \bibfield  {author} {\bibinfo {author} {\bibfnamefont {S.}~\bibnamefont
  {Plyasunov}}\ and\ \bibinfo {author} {\bibfnamefont {A.~P.}\ \bibnamefont
  {Arkin}},\ }\bibfield  {title} {\enquote {\bibinfo {title} {Efficient
  stochastic sensitivity analysis of discrete event systems},}\ }\href@noop {}
  {\bibfield  {journal} {\bibinfo  {journal} {J. Comp. Phys.},\ }\textbf
  {\bibinfo {volume} {221}},\ \bibinfo {pages} {724--738} (\bibinfo {year}
  {2007})}\BibitemShut {NoStop}%
\bibitem [{\citenamefont {Kim}\ \emph {et~al.}(2007)\citenamefont {Kim},
  \citenamefont {Debusschere},\ and\ \citenamefont {Najm}}]{Kim:07}%
  \BibitemOpen
  \bibfield  {author} {\bibinfo {author} {\bibfnamefont {D.}~\bibnamefont
  {Kim}}, \bibinfo {author} {\bibfnamefont {B.}~\bibnamefont {Debusschere}}, \
  and\ \bibinfo {author} {\bibfnamefont {H.}~\bibnamefont {Najm}},\ }\bibfield
  {title} {\enquote {\bibinfo {title} {Spectral methods for parametric
  sensitivity in stochastic dynamical systems},}\ }\href@noop {} {\bibfield
  {journal} {\bibinfo  {journal} {Biophysical Journal},\ }\textbf {\bibinfo
  {volume} {92}},\ \bibinfo {pages} {379--393} (\bibinfo {year}
  {2007})}\BibitemShut {NoStop}%
\bibitem [{\citenamefont {Rathinam}\ \emph {et~al.}(2010)\citenamefont
  {Rathinam}, \citenamefont {Sheppard},\ and\ \citenamefont
  {Khammash}}]{Rathinam:10}%
  \BibitemOpen
  \bibfield  {author} {\bibinfo {author} {\bibfnamefont {M.}~\bibnamefont
  {Rathinam}}, \bibinfo {author} {\bibfnamefont {P.~W.}\ \bibnamefont
  {Sheppard}}, \ and\ \bibinfo {author} {\bibfnamefont {M.}~\bibnamefont
  {Khammash}},\ }\bibfield  {title} {\enquote {\bibinfo {title} {Efficient
  computation of parameter sensitivities of discrete stochastic chemical
  reaction networks},}\ }\href@noop {} {\bibfield  {journal} {\bibinfo
  {journal} {J. Chem. Phys.},\ }\textbf {\bibinfo {volume} {132}},\ \bibinfo
  {pages} {034103--(1--13)} (\bibinfo {year} {2010})}\BibitemShut {NoStop}%
\bibitem [{\citenamefont {Wu}\ \emph {et~al.}(2012)\citenamefont {Wu},
  \citenamefont {Schmidt}, \citenamefont {Wolverton},\ and\ \citenamefont
  {Schneider}}]{Wu:12}%
  \BibitemOpen
  \bibfield  {author} {\bibinfo {author} {\bibfnamefont {C.}~\bibnamefont
  {Wu}}, \bibinfo {author} {\bibfnamefont {D.~J.}\ \bibnamefont {Schmidt}},
  \bibinfo {author} {\bibfnamefont {C.}~\bibnamefont {Wolverton}}, \ and\
  \bibinfo {author} {\bibfnamefont {W.~F.}\ \bibnamefont {Schneider}},\
  }\bibfield  {title} {\enquote {\bibinfo {title} {{Accurate
  coverage-dependence incorporated into first-principles kinetic models:
  Catalytic NO oxidation on Pt (111)}},}\ }\href@noop {} {\bibfield  {journal}
  {\bibinfo  {journal} {{J. Catalysis}},\ }\textbf {\bibinfo {volume}
  {{286}}},\ \bibinfo {pages} {{88--94}} (\bibinfo {year}
  {{2012}})}\BibitemShut {NoStop}%
\bibitem [{\citenamefont {Liu}\ \emph {et~al.}(2006)\citenamefont {Liu},
  \citenamefont {Chen},\ and\ \citenamefont {Sudjianto}}]{Liu:06}%
  \BibitemOpen
  \bibfield  {author} {\bibinfo {author} {\bibfnamefont {H.}~\bibnamefont
  {Liu}}, \bibinfo {author} {\bibfnamefont {W.}~\bibnamefont {Chen}}, \ and\
  \bibinfo {author} {\bibfnamefont {A.}~\bibnamefont {Sudjianto}},\ }\bibfield
  {title} {\enquote {\bibinfo {title} {{Relative entropy based method for
  probabilistic sensitivity analysis in engineering design}},}\ }\href@noop {}
  {\bibfield  {journal} {\bibinfo  {journal} {{J. Mechanical Design}},\
  }\textbf {\bibinfo {volume} {{128}}},\ \bibinfo {pages} {{326--336}}
  (\bibinfo {year} {{2006}})}\BibitemShut {NoStop}%
\bibitem [{\citenamefont {{N. L\"{u}dtke and S. Panzeri and M. Brown and D. S.
  Broomhead and J. Knowles and M. A. Montemurro and D. B.
  Kell}}(2008)}]{Ludtke:08}%
  \BibitemOpen
  \bibfield  {author} {\bibinfo {author} {\bibnamefont {{N. L\"{u}dtke and S.
  Panzeri and M. Brown and D. S. Broomhead and J. Knowles and M. A. Montemurro
  and D. B. Kell}}},\ }\bibfield  {title} {\enquote {\bibinfo {title}
  {Information-theoretic sensitivity analysis: a general method for credit
  assignment in complex networks},}\ }\href@noop {} {\bibfield  {journal}
  {\bibinfo  {journal} {J. R. Soc. Interface},\ }\textbf {\bibinfo {volume}
  {5}},\ \bibinfo {pages} {223--235} (\bibinfo {year} {2008})}\BibitemShut
  {NoStop}%
\bibitem [{\citenamefont {Majda}\ and\ \citenamefont
  {Gershgorin}(2010)}]{Majda:10}%
  \BibitemOpen
  \bibfield  {author} {\bibinfo {author} {\bibfnamefont {A.~J.}\ \bibnamefont
  {Majda}}\ and\ \bibinfo {author} {\bibfnamefont {B.}~\bibnamefont
  {Gershgorin}},\ }\bibfield  {title} {\enquote {\bibinfo {title} {{Quantifying
  uncertainty in climate change science through empirical information
  theory}},}\ }\href@noop {} {\bibfield  {journal} {\bibinfo  {journal} {{Proc.
  of the National Academy of Sciences}},\ }\textbf {\bibinfo {volume}
  {{107}}},\ \bibinfo {pages} {{14958--14963}} (\bibinfo {year}
  {{2010}})}\BibitemShut {NoStop}%
\bibitem [{\citenamefont {Majda}\ and\ \citenamefont
  {Gershgorin}(2011)}]{Majda:11}%
  \BibitemOpen
  \bibfield  {author} {\bibinfo {author} {\bibfnamefont {A.~J.}\ \bibnamefont
  {Majda}}\ and\ \bibinfo {author} {\bibfnamefont {B.}~\bibnamefont
  {Gershgorin}},\ }\bibfield  {title} {\enquote {\bibinfo {title} {{Improving
  model fidelity and sensitivity for complex systems through empirical
  information theory}},}\ }\href@noop {} {\bibfield  {journal} {\bibinfo
  {journal} {{Proc. of the National Academy of Sciences}},\ }\textbf {\bibinfo
  {volume} {{108}}},\ \bibinfo {pages} {{10044--10049}} (\bibinfo {year}
  {{2011}})}\BibitemShut {NoStop}%
\bibitem [{\citenamefont {Komorowski}\ \emph {et~al.}(2011)\citenamefont
  {Komorowski}, \citenamefont {Costa}, \citenamefont {Rand},\ and\
  \citenamefont {Stumpf}}]{Komorowski:11}%
  \BibitemOpen
  \bibfield  {author} {\bibinfo {author} {\bibfnamefont {M.}~\bibnamefont
  {Komorowski}}, \bibinfo {author} {\bibfnamefont {M.~J.}\ \bibnamefont
  {Costa}}, \bibinfo {author} {\bibfnamefont {D.~A.}\ \bibnamefont {Rand}}, \
  and\ \bibinfo {author} {\bibfnamefont {M.~P.~H.}\ \bibnamefont {Stumpf}},\
  }\bibfield  {title} {\enquote {\bibinfo {title} {{Sensitivity, robustness,
  and identifiability in stochastic chemical kinetics models}},}\ }\href@noop
  {} {\bibfield  {journal} {\bibinfo  {journal} {{Proc. Natl. Acad. Sci.
  USA}},\ }\textbf {\bibinfo {volume} {{108}}},\ \bibinfo {pages}
  {{8645--8650}} (\bibinfo {year} {{2011}})}\BibitemShut {NoStop}%
\bibitem [{\citenamefont {Doering}\ \emph {et~al.}(2007)\citenamefont
  {Doering}, \citenamefont {Sargsyan}, \citenamefont {Sander},\ and\
  \citenamefont {Vanden-Eijnden}}]{Doering:07}%
  \BibitemOpen
  \bibfield  {author} {\bibinfo {author} {\bibfnamefont {C.~R.}\ \bibnamefont
  {Doering}}, \bibinfo {author} {\bibfnamefont {K.~V.}\ \bibnamefont
  {Sargsyan}}, \bibinfo {author} {\bibfnamefont {L.~M.}\ \bibnamefont
  {Sander}}, \ and\ \bibinfo {author} {\bibfnamefont {E.}~\bibnamefont
  {Vanden-Eijnden}},\ }\bibfield  {title} {\enquote {\bibinfo {title}
  {Asymptotics of rare events in birth--death processes bypassing the exact
  solutions},}\ }\href@noop {} {\bibfield  {journal} {\bibinfo  {journal}
  {Journal of Physics: Condensed Matter},\ }\textbf {\bibinfo {volume} {19}},\
  \bibinfo {pages} {065145--(1--12)} (\bibinfo {year} {2007})}\BibitemShut
  {NoStop}%
\bibitem [{\citenamefont {Hanggi}\ \emph {et~al.}(1984)\citenamefont {Hanggi},
  \citenamefont {Grabert}, \citenamefont {Talkner},\ and\ \citenamefont
  {Thomas}}]{Hanggi:84}%
  \BibitemOpen
  \bibfield  {author} {\bibinfo {author} {\bibfnamefont {P.}~\bibnamefont
  {Hanggi}}, \bibinfo {author} {\bibfnamefont {H.}~\bibnamefont {Grabert}},
  \bibinfo {author} {\bibfnamefont {P.}~\bibnamefont {Talkner}}, \ and\
  \bibinfo {author} {\bibfnamefont {H.}~\bibnamefont {Thomas}},\ }\bibfield
  {title} {\enquote {\bibinfo {title} {Bistable systems: {M}aster equation
  versus {F}okker-{P}lanck modeling},}\ }\href@noop {} {\bibfield  {journal}
  {\bibinfo  {journal} {Phys. Rev. A},\ }\textbf {\bibinfo {volume} {29}},\
  \bibinfo {pages} {371--378} (\bibinfo {year} {1984})}\BibitemShut {NoStop}%
\bibitem [{\citenamefont {Katsoulakis}\ \emph {et~al.}(2006)\citenamefont
  {Katsoulakis}, \citenamefont {Majda},\ and\ \citenamefont {Sopasakis}}]{KMS}%
  \BibitemOpen
  \bibfield  {author} {\bibinfo {author} {\bibfnamefont {M.~A.}\ \bibnamefont
  {Katsoulakis}}, \bibinfo {author} {\bibfnamefont {A.~J.}\ \bibnamefont
  {Majda}}, \ and\ \bibinfo {author} {\bibfnamefont {A.}~\bibnamefont
  {Sopasakis}},\ }\bibfield  {title} {\enquote {\bibinfo {title}
  {Intermittency, metastability and coarse graining for coupled
  deterministic-stochastic lattice systems},}\ }\href@noop {} {\bibfield
  {journal} {\bibinfo  {journal} {Nonlinearity},\ }\textbf {\bibinfo {volume}
  {19}},\ \bibinfo {pages} {1021--1047} (\bibinfo {year} {2006})}\BibitemShut
  {NoStop}%
\bibitem [{\citenamefont {Liggett}(1985)}]{Liggett:85}%
  \BibitemOpen
  \bibfield  {author} {\bibinfo {author} {\bibfnamefont {T.}~\bibnamefont
  {Liggett}},\ }\href@noop {} {\emph {\bibinfo {title} {Interacting particle
  systems}}}\ (\bibinfo  {publisher} {Springer - Berlin},\ \bibinfo {year}
  {1985})\BibitemShut {NoStop}%
\bibitem [{\citenamefont {Chatterjee}\ and\ \citenamefont
  {Vlachos}(2007)}]{Chatterjee:07}%
  \BibitemOpen
  \bibfield  {author} {\bibinfo {author} {\bibfnamefont {A.}~\bibnamefont
  {Chatterjee}}\ and\ \bibinfo {author} {\bibfnamefont {D.~G.}\ \bibnamefont
  {Vlachos}},\ }\bibfield  {title} {\enquote {\bibinfo {title} {An overview of
  spatial microscopic and accelerated kinetic {M}onte {C}arlo methods for
  materials' simulation},}\ }\href@noop {} {\bibfield  {journal} {\bibinfo
  {journal} {J. Computer-Aided Materials Design},\ }\textbf {\bibinfo {volume}
  {14}},\ \bibinfo {pages} {253--308} (\bibinfo {year} {2007})}\BibitemShut
  {NoStop}%
\bibitem [{\citenamefont {Cover}\ and\ \citenamefont
  {Thomas}(1991)}]{Cover:91}%
  \BibitemOpen
  \bibfield  {author} {\bibinfo {author} {\bibfnamefont {T.~M.}\ \bibnamefont
  {Cover}}\ and\ \bibinfo {author} {\bibfnamefont {J.~A.}\ \bibnamefont
  {Thomas}},\ }\href@noop {} {\emph {\bibinfo {title} {Elements of Information
  Theory}}}\ (\bibinfo  {publisher} {Wiley Series in Telecommunications},\
  \bibinfo {year} {1991})\BibitemShut {NoStop}%
\bibitem [{\citenamefont {Katsoulakis}\ and\ \citenamefont
  {Trashorras}(2006)}]{Kats:Trashorras:06}%
  \BibitemOpen
  \bibfield  {author} {\bibinfo {author} {\bibfnamefont {M.~A.}\ \bibnamefont
  {Katsoulakis}}\ and\ \bibinfo {author} {\bibfnamefont {J.}~\bibnamefont
  {Trashorras}},\ }\bibfield  {title} {\enquote {\bibinfo {title} {Information
  loss in coarse-graining of stochastic particle dynamics},}\ }\href@noop {}
  {\bibfield  {journal} {\bibinfo  {journal} {J. Stat. Phys.},\ }\textbf
  {\bibinfo {volume} {122}},\ \bibinfo {pages} {115--135} (\bibinfo {year}
  {2006})}\BibitemShut {NoStop}%
\bibitem [{\citenamefont {Katsoulakis}\ \emph {et~al.}(2007)\citenamefont
  {Katsoulakis}, \citenamefont {Rey-Bellet}, \citenamefont {Plech\'a\v{c}},\
  and\ \citenamefont {Tsagkarogiannis}}]{Kats:ReyBellet:07}%
  \BibitemOpen
  \bibfield  {author} {\bibinfo {author} {\bibfnamefont {M.~A.}\ \bibnamefont
  {Katsoulakis}}, \bibinfo {author} {\bibfnamefont {L.}~\bibnamefont
  {Rey-Bellet}}, \bibinfo {author} {\bibfnamefont {P.}~\bibnamefont
  {Plech\'a\v{c}}}, \ and\ \bibinfo {author} {\bibfnamefont {D.~K.}\
  \bibnamefont {Tsagkarogiannis}},\ }\bibfield  {title} {\enquote {\bibinfo
  {title} {Coarse-graining schemes and a posteriori error estimates for
  stochastic lattice systems},}\ }\href@noop {} {\bibfield  {journal} {\bibinfo
   {journal} {ESAIM-Math. Model. Num. Analysis},\ }\textbf {\bibinfo {volume}
  {41}},\ \bibinfo {pages} {627--660} (\bibinfo {year} {2007})}\BibitemShut
  {NoStop}%
\bibitem [{\citenamefont {Arnst}\ and\ \citenamefont
  {Ghanem}(2008)}]{Arnst:08}%
  \BibitemOpen
  \bibfield  {author} {\bibinfo {author} {\bibfnamefont {M.}~\bibnamefont
  {Arnst}}\ and\ \bibinfo {author} {\bibfnamefont {R.}~\bibnamefont {Ghanem}},\
  }\bibfield  {title} {\enquote {\bibinfo {title} {{Probabilistic equivalence
  and stochastic model reduction in multiscale analysis}},}\ }\href@noop {}
  {\bibfield  {journal} {\bibinfo  {journal} {{Comp. methods in applied mech.
  and eng.}},\ }\textbf {\bibinfo {volume} {{197}}},\ \bibinfo {pages}
  {{3584--3592}} (\bibinfo {year} {{2008}})}\BibitemShut {NoStop}%
\bibitem [{\citenamefont {Kipnis}\ and\ \citenamefont
  {Landim}(1999)}]{Kipnis:99}%
  \BibitemOpen
  \bibfield  {author} {\bibinfo {author} {\bibfnamefont {C.}~\bibnamefont
  {Kipnis}}\ and\ \bibinfo {author} {\bibfnamefont {C.}~\bibnamefont
  {Landim}},\ }\href@noop {} {\emph {\bibinfo {title} {Scaling Limits of
  Interacting Particle Systems}}}\ (\bibinfo  {publisher} {Springer-Verlag},\
  \bibinfo {year} {1999})\BibitemShut {NoStop}%
\bibitem [{\citenamefont {Maes}\ \emph {et~al.}(2000)\citenamefont {Maes},
  \citenamefont {Redig},\ and\ \citenamefont {Moffaert}}]{Maes:00}%
  \BibitemOpen
  \bibfield  {author} {\bibinfo {author} {\bibfnamefont {C.}~\bibnamefont
  {Maes}}, \bibinfo {author} {\bibfnamefont {F.}~\bibnamefont {Redig}}, \ and\
  \bibinfo {author} {\bibfnamefont {A.~V.}\ \bibnamefont {Moffaert}},\
  }\bibfield  {title} {\enquote {\bibinfo {title} {On the definition of entropy
  production, via examples},}\ }\href@noop {} {\bibfield  {journal} {\bibinfo
  {journal} {J. Math. Phys.},\ }\textbf {\bibinfo {volume} {41}},\ \bibinfo
  {pages} {1528--1553} (\bibinfo {year} {2000})}\BibitemShut {NoStop}%
\bibitem [{\citenamefont {Abramov}\ \emph {et~al.}(2005)\citenamefont
  {Abramov}, \citenamefont {Grote},\ and\ \citenamefont {Majda}}]{Abramov:05}%
  \BibitemOpen
  \bibfield  {author} {\bibinfo {author} {\bibfnamefont {R.~V.}\ \bibnamefont
  {Abramov}}, \bibinfo {author} {\bibfnamefont {M.~J.}\ \bibnamefont {Grote}},
  \ and\ \bibinfo {author} {\bibfnamefont {A.~J.}\ \bibnamefont {Majda}},\
  }\href@noop {} {\emph {\bibinfo {title} {Information Theory and Stochastics
  for Multiscale Nonlinear Systems}}}\ (\bibinfo  {publisher} {CRM Monograph
  Series},\ \bibinfo {year} {2005})\BibitemShut {NoStop}%
\bibitem [{\citenamefont {Rao}\ \emph {et~al.}(2009)\citenamefont {Rao},
  \citenamefont {Imam}, \citenamefont {Ramanathan},\ and\ \citenamefont
  {Pushpavanam}}]{Rao:10}%
  \BibitemOpen
  \bibfield  {author} {\bibinfo {author} {\bibfnamefont {S.~K.}\ \bibnamefont
  {Rao}}, \bibinfo {author} {\bibfnamefont {R.}~\bibnamefont {Imam}}, \bibinfo
  {author} {\bibfnamefont {K.}~\bibnamefont {Ramanathan}}, \ and\ \bibinfo
  {author} {\bibfnamefont {S.}~\bibnamefont {Pushpavanam}},\ }\bibfield
  {title} {\enquote {\bibinfo {title} {{Sensitivity Analysis and Kinetic
  Parameter Estimation in a Three Way Catalytic Converter}},}\ }\href@noop {}
  {\bibfield  {journal} {\bibinfo  {journal} {{Industrial \& Engineering
  Chemistry Research}},\ }\textbf {\bibinfo {volume} {{48}}},\ \bibinfo {pages}
  {{3779--3790}} (\bibinfo {year} {{2009}})}\BibitemShut {NoStop}%
\bibitem [{\citenamefont {Braatz}\ \emph {et~al.}(2006)\citenamefont {Braatz},
  \citenamefont {Alkire}, \citenamefont {Seebauer}, \citenamefont {Rusli},
  \citenamefont {Gunawan}, \citenamefont {Drews}, \citenamefont {Li},\ and\
  \citenamefont {He}}]{Braatz:06}%
  \BibitemOpen
  \bibfield  {author} {\bibinfo {author} {\bibfnamefont {R.}~\bibnamefont
  {Braatz}}, \bibinfo {author} {\bibfnamefont {R.}~\bibnamefont {Alkire}},
  \bibinfo {author} {\bibfnamefont {E.}~\bibnamefont {Seebauer}}, \bibinfo
  {author} {\bibfnamefont {E.}~\bibnamefont {Rusli}}, \bibinfo {author}
  {\bibfnamefont {R.}~\bibnamefont {Gunawan}}, \bibinfo {author} {\bibfnamefont
  {T.}~\bibnamefont {Drews}}, \bibinfo {author} {\bibfnamefont
  {X.}~\bibnamefont {Li}}, \ and\ \bibinfo {author} {\bibfnamefont
  {Y.}~\bibnamefont {He}},\ }\bibfield  {title} {\enquote {\bibinfo {title}
  {Perspectives on the design and control of multiscale systems},}\ }\href@noop
  {} {\bibfield  {journal} {\bibinfo  {journal} {J. Proc. Control},\ }\textbf
  {\bibinfo {volume} {16}},\ \bibinfo {pages} {193--204} (\bibinfo {year}
  {2006})}\BibitemShut {NoStop}%
\bibitem [{\citenamefont {Rothenberg}(1971)}]{Rothenberg:71}%
  \BibitemOpen
  \bibfield  {author} {\bibinfo {author} {\bibfnamefont {T.}~\bibnamefont
  {Rothenberg}},\ }\bibfield  {title} {\enquote {\bibinfo {title}
  {{Identification in parametric models}},}\ }\href@noop {} {\bibfield
  {journal} {\bibinfo  {journal} {{ECONOMETRICA}},\ }\textbf {\bibinfo {volume}
  {{39}}},\ \bibinfo {pages} {{577--0591}} (\bibinfo {year}
  {{1971}})}\BibitemShut {NoStop}%
\bibitem [{\citenamefont {Emery}\ and\ \citenamefont
  {Nenarokomov}(1998)}]{Emery:98}%
  \BibitemOpen
  \bibfield  {author} {\bibinfo {author} {\bibfnamefont {A.~F.}\ \bibnamefont
  {Emery}}\ and\ \bibinfo {author} {\bibfnamefont {A.~V.}\ \bibnamefont
  {Nenarokomov}},\ }\bibfield  {title} {\enquote {\bibinfo {title} {{Optimal
  experiment design}},}\ }\href@noop {} {\bibfield  {journal} {\bibinfo
  {journal} {{Measurement Science \& Technology}},\ }\textbf {\bibinfo {volume}
  {{9}}},\ \bibinfo {pages} {{864--876}} (\bibinfo {year}
  {{1998}})}\BibitemShut {NoStop}%
\bibitem [{\citenamefont {Prasad}\ \emph {et~al.}(2010)\citenamefont {Prasad},
  \citenamefont {Karim}, \citenamefont {Ulissi}, \citenamefont {Zagrobelny},\
  and\ \citenamefont {Vlachos}}]{Prasad:10}%
  \BibitemOpen
  \bibfield  {author} {\bibinfo {author} {\bibfnamefont {V.}~\bibnamefont
  {Prasad}}, \bibinfo {author} {\bibfnamefont {A.~M.}\ \bibnamefont {Karim}},
  \bibinfo {author} {\bibfnamefont {Z.}~\bibnamefont {Ulissi}}, \bibinfo
  {author} {\bibfnamefont {M.}~\bibnamefont {Zagrobelny}}, \ and\ \bibinfo
  {author} {\bibfnamefont {D.~G.}\ \bibnamefont {Vlachos}},\ }\bibfield
  {title} {\enquote {\bibinfo {title} {{High throughput multiscale modeling for
  design of experiments, catalysts, and reactors: Application to hydrogen
  production from ammonia}},}\ }\href@noop {} {\bibfield  {journal} {\bibinfo
  {journal} {{Chem. Eng. Sci.}},\ }\textbf {\bibinfo {volume} {{65}}},\
  \bibinfo {pages} {{240--246}} (\bibinfo {year} {{2010}})}\BibitemShut
  {NoStop}%
\bibitem [{\citenamefont {Gillespie}(2001)}]{Gillespie:01}%
  \BibitemOpen
  \bibfield  {author} {\bibinfo {author} {\bibfnamefont {D.~T.}\ \bibnamefont
  {Gillespie}},\ }\bibfield  {title} {\enquote {\bibinfo {title} {Approximated
  accelerated stochastic simulation of chemically reacting systems},}\
  }\href@noop {} {\bibfield  {journal} {\bibinfo  {journal} {J. Chem. Phys.},\
  }\textbf {\bibinfo {volume} {115}},\ \bibinfo {pages} {1716--1733} (\bibinfo
  {year} {2001})}\BibitemShut {NoStop}%
\bibitem [{\citenamefont {Plimpton}\ \emph {et~al.}(2009)\citenamefont
  {Plimpton}, \citenamefont {Battaile}, \citenamefont {Chandross},
  \citenamefont {Holm}, \citenamefont {Thompson}, \citenamefont {Tikare},
  \citenamefont {Wagner}, \citenamefont {Webb}, \citenamefont {Zhou},
  \citenamefont {Cardona},\ and\ \citenamefont {Slepoy}}]{SPPARKS}%
  \BibitemOpen
  \bibfield  {author} {\bibinfo {author} {\bibfnamefont {S.}~\bibnamefont
  {Plimpton}}, \bibinfo {author} {\bibfnamefont {C.}~\bibnamefont {Battaile}},
  \bibinfo {author} {\bibfnamefont {M.}~\bibnamefont {Chandross}}, \bibinfo
  {author} {\bibfnamefont {L.}~\bibnamefont {Holm}}, \bibinfo {author}
  {\bibfnamefont {A.}~\bibnamefont {Thompson}}, \bibinfo {author}
  {\bibfnamefont {V.}~\bibnamefont {Tikare}}, \bibinfo {author} {\bibfnamefont
  {G.}~\bibnamefont {Wagner}}, \bibinfo {author} {\bibfnamefont
  {E.}~\bibnamefont {Webb}}, \bibinfo {author} {\bibfnamefont {X.}~\bibnamefont
  {Zhou}}, \bibinfo {author} {\bibfnamefont {C.~G.}\ \bibnamefont {Cardona}}, \
  and\ \bibinfo {author} {\bibfnamefont {A.}~\bibnamefont {Slepoy}},\
  }\href@noop {} {\enquote {\bibinfo {title} {{Crossing the Mesoscale No-Man's
  Land via Parallel Kinetic Monte Carlo}},}\ }\bibinfo {type} {Tech. Rep.}\
  (\bibinfo  {institution} {Sandia National Laboratory},\ \bibinfo {year}
  {2009})\BibitemShut {NoStop}%
\bibitem [{\citenamefont {Plimpton}\ and\ \citenamefont
  {Thompson}(2012)}]{LAMMPS}%
  \BibitemOpen
  \bibfield  {author} {\bibinfo {author} {\bibfnamefont {S.~J.}\ \bibnamefont
  {Plimpton}}\ and\ \bibinfo {author} {\bibfnamefont {A.~P.}\ \bibnamefont
  {Thompson}},\ }\bibfield  {title} {\enquote {\bibinfo {title} {{Computational
  aspects of many-body potentials}},}\ }\href@noop {} {\bibfield  {journal}
  {\bibinfo  {journal} {{MRS Bull.}},\ }\textbf {\bibinfo {volume} {{37}}},\
  \bibinfo {pages} {{513--521}} (\bibinfo {year} {{2012}})}\BibitemShut
  {NoStop}%
\bibitem [{\citenamefont {{Arampatzis}}\ \emph {et~al.}(2012)\citenamefont
  {{Arampatzis}}, \citenamefont {{Katsoulakis}}, \citenamefont {{Plechac}},
  \citenamefont {{Taufer}},\ and\ \citenamefont {{Xu}}}]{Arampatzis:12}%
  \BibitemOpen
  \bibfield  {author} {\bibinfo {author} {\bibfnamefont {G.}~\bibnamefont
  {{Arampatzis}}}, \bibinfo {author} {\bibfnamefont {M.~A.}\ \bibnamefont
  {{Katsoulakis}}}, \bibinfo {author} {\bibfnamefont {P.}~\bibnamefont
  {{Plechac}}}, \bibinfo {author} {\bibfnamefont {M.}~\bibnamefont {{Taufer}}},
  \ and\ \bibinfo {author} {\bibfnamefont {L.}~\bibnamefont {{Xu}}},\
  }\bibfield  {title} {\enquote {\bibinfo {title} {{Hierarchical
  fractional-step approximations and parallel kinetic {M}onte {C}arlo
  algorithms}},}\ }\href@noop {} {\bibfield  {journal} {\bibinfo  {journal} {J.
  Comp. Phys.},\ \bibinfo {pages} {7795--7814}} (\bibinfo {year}
  {2012})}\BibitemShut {NoStop}%
\bibitem [{\citenamefont {Hansen}\ and\ \citenamefont
  {Neurock}(2000)}]{Hansen:00}%
  \BibitemOpen
  \bibfield  {author} {\bibinfo {author} {\bibfnamefont {E.~W.}\ \bibnamefont
  {Hansen}}\ and\ \bibinfo {author} {\bibfnamefont {M.}~\bibnamefont
  {Neurock}},\ }\bibfield  {title} {\enquote {\bibinfo {title}
  {First-principles-based {Monte Carlo} simulation of ethylene hydrogenation
  kinetics on {Pd}},}\ }\href@noop {} {\bibfield  {journal} {\bibinfo
  {journal} {J. Catalysis},\ }\textbf {\bibinfo {volume} {196}},\ \bibinfo
  {pages} {241--252} (\bibinfo {year} {2000})}\BibitemShut {NoStop}%
\bibitem [{\citenamefont {Meskine}\ \emph {et~al.}(2009)\citenamefont
  {Meskine}, \citenamefont {Matera}, \citenamefont {Scheffler}, \citenamefont
  {Reuter},\ and\ \citenamefont {Metiu}}]{Meskine:09}%
  \BibitemOpen
  \bibfield  {author} {\bibinfo {author} {\bibfnamefont {H.}~\bibnamefont
  {Meskine}}, \bibinfo {author} {\bibfnamefont {S.}~\bibnamefont {Matera}},
  \bibinfo {author} {\bibfnamefont {M.}~\bibnamefont {Scheffler}}, \bibinfo
  {author} {\bibfnamefont {K.}~\bibnamefont {Reuter}}, \ and\ \bibinfo {author}
  {\bibfnamefont {H.}~\bibnamefont {Metiu}},\ }\bibfield  {title} {\enquote
  {\bibinfo {title} {Examination of the concept of degree of rate control by
  first-principles kinetic {Monte Carlo} simulations},}\ }\href@noop {}
  {\bibfield  {journal} {\bibinfo  {journal} {Surf. Science},\ }\textbf
  {\bibinfo {volume} {603(10-12)}},\ \bibinfo {pages} {1724--1730} (\bibinfo
  {year} {2009})}\BibitemShut {NoStop}%
\bibitem [{\citenamefont {Stamatakis}\ \emph {et~al.}(2011)\citenamefont
  {Stamatakis}, \citenamefont {Chen},\ and\ \citenamefont
  {Vlachos}}]{Stamatakis:11}%
  \BibitemOpen
  \bibfield  {author} {\bibinfo {author} {\bibfnamefont {M.}~\bibnamefont
  {Stamatakis}}, \bibinfo {author} {\bibfnamefont {Y.}~\bibnamefont {Chen}}, \
  and\ \bibinfo {author} {\bibfnamefont {D.~G.}\ \bibnamefont {Vlachos}},\
  }\bibfield  {title} {\enquote {\bibinfo {title} {First-principles-based
  kinetic {M}onte {C}arlo simulation of the structure sensitivity of the
  water-gas shift reaction on {P}latinum surfaces},}\ }\href@noop {} {\bibfield
   {journal} {\bibinfo  {journal} {Journal of Physical Chemistry C},\ }\textbf
  {\bibinfo {volume} {115(50)}},\ \bibinfo {pages} {24750--24762} (\bibinfo
  {year} {2011})}\BibitemShut {NoStop}%
\bibitem [{Note1()}]{Note1}%
  \BibitemOpen
  \bibinfo {note} {This Lebesgue continuity assumption is merely for
  simplification purposes and it can be easily generalized.}\BibitemShut
  {Stop}%
\bibitem [{\citenamefont {Liptser}\ and\ \citenamefont
  {Shiryaev}(1977)}]{Liptser:77}%
  \BibitemOpen
  \bibfield  {author} {\bibinfo {author} {\bibfnamefont {R.~S.}\ \bibnamefont
  {Liptser}}\ and\ \bibinfo {author} {\bibfnamefont {A.~N.}\ \bibnamefont
  {Shiryaev}},\ }\href@noop {} {\emph {\bibinfo {title} {Statistics of Random
  Processes: {I \& II}}}}\ (\bibinfo  {publisher} {Springer},\ \bibinfo {year}
  {1977})\BibitemShut {NoStop}%
\bibitem [{\citenamefont {Dumitrescu}(1988)}]{Dumitrescu:88}%
  \BibitemOpen
  \bibfield  {author} {\bibinfo {author} {\bibfnamefont {M.~E.}\ \bibnamefont
  {Dumitrescu}},\ }\bibfield  {title} {\enquote {\bibinfo {title} {Some
  informational properties of markov pure-jump processes},}\ }\href@noop {}
  {\bibfield  {journal} {\bibinfo  {journal} {C. P. Matematiky},\ }\textbf
  {\bibinfo {volume} {113}},\ \bibinfo {pages} {429--434} (\bibinfo {year}
  {1988})}\BibitemShut {NoStop}%
\bibitem [{\citenamefont {Wen}\ and\ \citenamefont {Weiguo}(1996)}]{Wen:96}%
  \BibitemOpen
  \bibfield  {author} {\bibinfo {author} {\bibfnamefont {L.}~\bibnamefont
  {Wen}}\ and\ \bibinfo {author} {\bibfnamefont {Y.}~\bibnamefont {Weiguo}},\
  }\bibfield  {title} {\enquote {\bibinfo {title} {An extension of
  {Shannon-McMillan} theorem and some limit properties for nonhomogeneous
  {M}arkov chains},}\ }\href@noop {} {\bibfield  {journal} {\bibinfo  {journal}
  {Stochastic Processes and their Applications},\ }\textbf {\bibinfo {volume}
  {61}},\ \bibinfo {pages} {129--145} (\bibinfo {year} {1996})}\BibitemShut
  {NoStop}%
\bibitem [{\citenamefont {Limnios}\ and\ \citenamefont
  {Oprisan}(2001)}]{Limnios:01}%
  \BibitemOpen
  \bibfield  {author} {\bibinfo {author} {\bibfnamefont {N.}~\bibnamefont
  {Limnios}}\ and\ \bibinfo {author} {\bibfnamefont {G.}~\bibnamefont
  {Oprisan}},\ }\href@noop {} {\emph {\bibinfo {title} {Semi-Markov Processes
  and Reliability}}}\ (\bibinfo  {publisher} {Springer},\ \bibinfo {year}
  {2001})\BibitemShut {NoStop}%
\bibitem [{\citenamefont {Lutz}\ and\ \citenamefont
  {Kiremidjian}(1993)}]{Lutz:93}%
  \BibitemOpen
  \bibfield  {author} {\bibinfo {author} {\bibfnamefont {K.~A.}\ \bibnamefont
  {Lutz}}\ and\ \bibinfo {author} {\bibfnamefont {A.~S.}\ \bibnamefont
  {Kiremidjian}},\ }\href@noop {} {\enquote {\bibinfo {title} {A generalized
  semi-{M}arkov process for modeling spatially and temporally dependent
  earthquakes},}\ }\bibinfo {type} {Tech. Rep.}\ (\bibinfo  {institution} {The
  J. A. Blume Earthquake Engineering Center},\ \bibinfo {year}
  {1993})\BibitemShut {NoStop}%
\bibitem [{\citenamefont {Janssen}\ and\ \citenamefont
  {Manca}(2006)}]{Janssen:06}%
  \BibitemOpen
  \bibfield  {author} {\bibinfo {author} {\bibfnamefont {J.}~\bibnamefont
  {Janssen}}\ and\ \bibinfo {author} {\bibfnamefont {R.}~\bibnamefont
  {Manca}},\ }\href@noop {} {\emph {\bibinfo {title} {Applied Semi-Markov
  Processes}}}\ (\bibinfo  {publisher} {Springer},\ \bibinfo {year}
  {2006})\BibitemShut {NoStop}%
\bibitem [{\citenamefont {Girardin}\ and\ \citenamefont
  {Limnios}(2003)}]{Girardin:03}%
  \BibitemOpen
  \bibfield  {author} {\bibinfo {author} {\bibfnamefont {V.}~\bibnamefont
  {Girardin}}\ and\ \bibinfo {author} {\bibfnamefont {N.}~\bibnamefont
  {Limnios}},\ }\bibfield  {title} {\enquote {\bibinfo {title} {On the entropy
  for semi-{M}arkov processes},}\ }\href@noop {} {\bibfield  {journal}
  {\bibinfo  {journal} {J. Appl. Probab.},\ }\textbf {\bibinfo {volume} {40}},\
  \bibinfo {pages} {1060--1068} (\bibinfo {year} {2003})}\BibitemShut {NoStop}%
\bibitem [{\citenamefont {Gillespie}(1976)}]{Gillespie:76}%
  \BibitemOpen
  \bibfield  {author} {\bibinfo {author} {\bibfnamefont {D.~T.}\ \bibnamefont
  {Gillespie}},\ }\bibfield  {title} {\enquote {\bibinfo {title} {A general
  method for numerically simulating the stochastic time evolution of coupled
  chemical reactions},}\ }\href@noop {} {\bibfield  {journal} {\bibinfo
  {journal} {J. Comp. Phys.},\ }\textbf {\bibinfo {volume} {22}},\ \bibinfo
  {pages} {403--434} (\bibinfo {year} {1976})}\BibitemShut {NoStop}%
\bibitem [{\citenamefont {{Schl\"{o}gl}}(1972)}]{Schlogl:72}%
  \BibitemOpen
  \bibfield  {author} {\bibinfo {author} {\bibfnamefont {F.}~\bibnamefont
  {{Schl\"{o}gl}}},\ }\bibfield  {title} {\enquote {\bibinfo {title} {Chemical
  reaction models for nonequilibrium phase transition},}\ }\href@noop {}
  {\bibfield  {journal} {\bibinfo  {journal} {Z. Physik},\ }\textbf {\bibinfo
  {volume} {253}},\ \bibinfo {pages} {147--161} (\bibinfo {year}
  {1972})}\BibitemShut {NoStop}%
\bibitem [{\citenamefont {Vellela}\ and\ \citenamefont
  {Qian}(2009)}]{Vellela:09}%
  \BibitemOpen
  \bibfield  {author} {\bibinfo {author} {\bibfnamefont {M.}~\bibnamefont
  {Vellela}}\ and\ \bibinfo {author} {\bibfnamefont {H.}~\bibnamefont {Qian}},\
  }\bibfield  {title} {\enquote {\bibinfo {title} {Stochastic dynamics and
  non-equilibrium thermodynamics of a bistable chemical system: the
  {Schl\"{o}gl} model revisited},}\ }\href@noop {} {\bibfield  {journal}
  {\bibinfo  {journal} {J. R. Soc. Interface},\ }\textbf {\bibinfo {volume}
  {6}},\ \bibinfo {pages} {925--940} (\bibinfo {year} {2009})}\BibitemShut
  {NoStop}%
\bibitem [{\citenamefont {Gardiner}(1985)}]{Gardiner:85}%
  \BibitemOpen
  \bibfield  {author} {\bibinfo {author} {\bibfnamefont {C.}~\bibnamefont
  {Gardiner}},\ }\href@noop {} {\emph {\bibinfo {title} {Handbook of Stochastic
  Methods: for Physics, Chemistry and the Natural Sciences}}}\ (\bibinfo
  {publisher} {Springer},\ \bibinfo {year} {1985})\BibitemShut {NoStop}%
\bibitem [{\citenamefont {Degasperi}\ and\ \citenamefont
  {Gilmore}(2008)}]{Degasperi:08}%
  \BibitemOpen
  \bibfield  {author} {\bibinfo {author} {\bibfnamefont {A.}~\bibnamefont
  {Degasperi}}\ and\ \bibinfo {author} {\bibfnamefont {S.}~\bibnamefont
  {Gilmore}},\ }\bibfield  {title} {\enquote {\bibinfo {title} {Sensitivity
  analysis of stochastic models of bistable biochemical reactions},}\
  }\href@noop {} {\bibfield  {journal} {\bibinfo  {journal} {SFM 2008},\
  \bibinfo {pages} {1--20}} (\bibinfo {year} {2008})}\BibitemShut {NoStop}%
\bibitem [{\citenamefont {Rapaport}(1995)}]{Rapaport:95}%
  \BibitemOpen
  \bibfield  {author} {\bibinfo {author} {\bibfnamefont {D.~C.}\ \bibnamefont
  {Rapaport}},\ }\href@noop {} {\emph {\bibinfo {title} {The Art of Molecular
  Dynamics Simulations}}}\ (\bibinfo  {publisher} {Cambridge University
  Press},\ \bibinfo {year} {1995})\BibitemShut {NoStop}%
\bibitem [{\citenamefont {Schlick}(2002)}]{Schlick:02}%
  \BibitemOpen
  \bibfield  {author} {\bibinfo {author} {\bibfnamefont {T.}~\bibnamefont
  {Schlick}},\ }\href@noop {} {\emph {\bibinfo {title} {Molecular Modeling and
  Simulation}}}\ (\bibinfo  {publisher} {Springer},\ \bibinfo {year}
  {2002})\BibitemShut {NoStop}%
\bibitem [{\citenamefont {Frenkel}\ and\ \citenamefont
  {Smit}(2002)}]{Frenkel:02}%
  \BibitemOpen
  \bibfield  {author} {\bibinfo {author} {\bibfnamefont {D.}~\bibnamefont
  {Frenkel}}\ and\ \bibinfo {author} {\bibfnamefont {B.}~\bibnamefont {Smit}},\
  }\href@noop {} {\emph {\bibinfo {title} {Understanding Molecular Simulation,
  From Algorithms to Applications}}}\ (\bibinfo  {publisher} {Academic Press},\
  \bibinfo {year} {2002})\BibitemShut {NoStop}%
\bibitem [{\citenamefont {Lelievre}\ \emph {et~al.}(2010)\citenamefont
  {Lelievre}, \citenamefont {Rousset},\ and\ \citenamefont
  {Stoltz}}]{Lelievre:10}%
  \BibitemOpen
  \bibfield  {author} {\bibinfo {author} {\bibfnamefont {T.}~\bibnamefont
  {Lelievre}}, \bibinfo {author} {\bibfnamefont {M.}~\bibnamefont {Rousset}}, \
  and\ \bibinfo {author} {\bibfnamefont {G.}~\bibnamefont {Stoltz}},\
  }\href@noop {} {\emph {\bibinfo {title} {Free energy computations: a
  mathematical perspective}}}\ (\bibinfo  {publisher} {Imperial College
  Press},\ \bibinfo {year} {2010})\BibitemShut {NoStop}%
\bibitem [{\citenamefont {Ebeling}\ and\ \citenamefont
  {Schimansky-Geier}(2008)}]{Ebeling:08}%
  \BibitemOpen
  \bibfield  {author} {\bibinfo {author} {\bibfnamefont {W.}~\bibnamefont
  {Ebeling}}\ and\ \bibinfo {author} {\bibfnamefont {L.}~\bibnamefont
  {Schimansky-Geier}},\ }\bibfield  {title} {\enquote {\bibinfo {title} {Swarm
  dynamics, attractors and bifurcations of active {B}rownian motion},}\
  }\href@noop {} {\bibfield  {journal} {\bibinfo  {journal} {Eur. Phys. J.
  Special Topics},\ }\textbf {\bibinfo {volume} {157}},\ \bibinfo {pages}
  {17--31} (\bibinfo {year} {2008})}\BibitemShut {NoStop}%
\bibitem [{\citenamefont {Scemama}\ \emph {et~al.}(2006)\citenamefont
  {Scemama}, \citenamefont {{Leli\`{e}vre}}, \citenamefont {Stoltz},
  \citenamefont {{Canc\`{e}s}},\ and\ \citenamefont {Caffarel}}]{Scemama:06}%
  \BibitemOpen
  \bibfield  {author} {\bibinfo {author} {\bibfnamefont {A.}~\bibnamefont
  {Scemama}}, \bibinfo {author} {\bibfnamefont {T.}~\bibnamefont
  {{Leli\`{e}vre}}}, \bibinfo {author} {\bibfnamefont {G.}~\bibnamefont
  {Stoltz}}, \bibinfo {author} {\bibfnamefont {E.}~\bibnamefont
  {{Canc\`{e}s}}}, \ and\ \bibinfo {author} {\bibfnamefont {M.}~\bibnamefont
  {Caffarel}},\ }\bibfield  {title} {\enquote {\bibinfo {title} {An efficient
  sampling algorithm for variational {Monte-Carlo}},}\ }\href@noop {}
  {\bibfield  {journal} {\bibinfo  {journal} {J. Chem. Phys.},\ }\textbf
  {\bibinfo {volume} {125}},\ \bibinfo {pages} {114105(1--9)} (\bibinfo {year}
  {2006})}\BibitemShut {NoStop}%
\bibitem [{\citenamefont {{Canc\`{e}s}}\ \emph {et~al.}(2007)\citenamefont
  {{Canc\`{e}s}}, \citenamefont {Legoll},\ and\ \citenamefont
  {Stoltz}}]{Cances:07}%
  \BibitemOpen
  \bibfield  {author} {\bibinfo {author} {\bibfnamefont {E.}~\bibnamefont
  {{Canc\`{e}s}}}, \bibinfo {author} {\bibfnamefont {F.}~\bibnamefont
  {Legoll}}, \ and\ \bibinfo {author} {\bibfnamefont {G.}~\bibnamefont
  {Stoltz}},\ }\bibfield  {title} {\enquote {\bibinfo {title} {Theoretical and
  numerical comparison of sampling methods for molecular dynamics},}\
  }\href@noop {} {\bibfield  {journal} {\bibinfo  {journal} {Math. Model.
  Numer. Anal.},\ }\textbf {\bibinfo {volume} {41}},\ \bibinfo {pages}
  {351--390} (\bibinfo {year} {2007})}\BibitemShut {NoStop}%
\bibitem [{\citenamefont {Kaplan}(2003)}]{Kaplan:03}%
  \BibitemOpen
  \bibfield  {author} {\bibinfo {author} {\bibfnamefont {I.}~\bibnamefont
  {Kaplan}},\ }\href@noop {} {\emph {\bibinfo {title} {Handbook of Molecular
  Physics and Quantum Chemistry}}}\ (\bibinfo  {publisher} {Wiley},\ \bibinfo
  {year} {2003})\BibitemShut {NoStop}%
\bibitem [{\citenamefont {Lebowitz}\ and\ \citenamefont
  {Spohn}(1999)}]{Lebowitz:99}%
  \BibitemOpen
  \bibfield  {author} {\bibinfo {author} {\bibfnamefont {J.~L.}\ \bibnamefont
  {Lebowitz}}\ and\ \bibinfo {author} {\bibfnamefont {H.}~\bibnamefont
  {Spohn}},\ }\bibfield  {title} {\enquote {\bibinfo {title} {A
  {G}allavotti-{C}ohen type symmetry in the large deviation functional for
  stochastic dynamics},}\ }\href@noop {} {\bibfield  {journal} {\bibinfo
  {journal} {J. Stat. Phys.},\ }\textbf {\bibinfo {volume} {95}},\ \bibinfo
  {pages} {333--365} (\bibinfo {year} {1999})}\BibitemShut {NoStop}%
\bibitem [{\citenamefont {Sweet}\ \emph {et~al.}(2011)\citenamefont {Sweet},
  \citenamefont {Chatterjee}, \citenamefont {Xu}, \citenamefont {Bisordi},
  \citenamefont {Rosen},\ and\ \citenamefont {Alber}}]{Alber:11}%
  \BibitemOpen
  \bibfield  {author} {\bibinfo {author} {\bibfnamefont {C.~R.}\ \bibnamefont
  {Sweet}}, \bibinfo {author} {\bibfnamefont {S.}~\bibnamefont {Chatterjee}},
  \bibinfo {author} {\bibfnamefont {Z.}~\bibnamefont {Xu}}, \bibinfo {author}
  {\bibfnamefont {K.}~\bibnamefont {Bisordi}}, \bibinfo {author} {\bibfnamefont
  {E.~D.}\ \bibnamefont {Rosen}}, \ and\ \bibinfo {author} {\bibfnamefont
  {M.}~\bibnamefont {Alber}},\ }\bibfield  {title} {\enquote {\bibinfo {title}
  {{Modelling platelet-blood flow interaction using the subcellular element
  Langevin method}},}\ }\href@noop {} {\bibfield  {journal} {\bibinfo
  {journal} {{J. Royal Society Interface}},\ }\textbf {\bibinfo {volume}
  {{8}}},\ \bibinfo {pages} {{1760--1771}} (\bibinfo {year}
  {{2011}})}\BibitemShut {NoStop}%
\bibitem [{\citenamefont {Brunger}\ \emph {et~al.}(1984)\citenamefont
  {Brunger}, \citenamefont {Brooks},\ and\ \citenamefont
  {Karplus}}]{Brunger:84}%
  \BibitemOpen
  \bibfield  {author} {\bibinfo {author} {\bibfnamefont {A.}~\bibnamefont
  {Brunger}}, \bibinfo {author} {\bibfnamefont {C.~B.}\ \bibnamefont {Brooks}},
  \ and\ \bibinfo {author} {\bibfnamefont {M.}~\bibnamefont {Karplus}},\
  }\bibfield  {title} {\enquote {\bibinfo {title} {Stochastic boundary
  conditions for molecular dynamics simulations of {ST2} water},}\ }\href@noop
  {} {\bibfield  {journal} {\bibinfo  {journal} {Chem. Phys. Lett.},\ }\textbf
  {\bibinfo {volume} {105}},\ \bibinfo {pages} {495--500} (\bibinfo {year}
  {1984})}\BibitemShut {NoStop}%
\bibitem [{\citenamefont {Ziff}\ \emph {et~al.}(1986)\citenamefont {Ziff},
  \citenamefont {Gulari},\ and\ \citenamefont {Barshad}}]{Ziff:86}%
  \BibitemOpen
  \bibfield  {author} {\bibinfo {author} {\bibfnamefont {R.~M.}\ \bibnamefont
  {Ziff}}, \bibinfo {author} {\bibfnamefont {E.}~\bibnamefont {Gulari}}, \ and\
  \bibinfo {author} {\bibfnamefont {Y.}~\bibnamefont {Barshad}},\ }\bibfield
  {title} {\enquote {\bibinfo {title} {Kinetic phase transitions in an
  irreversible surface-reaction model},}\ }\href@noop {} {\bibfield  {journal}
  {\bibinfo  {journal} {Phys. Rev. Lett.},\ }\textbf {\bibinfo {volume} {56}},\
  \bibinfo {pages} {2553} (\bibinfo {year} {1986})}\BibitemShut {NoStop}%
\bibitem [{\citenamefont {Chan}\ \emph {et~al.}(1997)\citenamefont {Chan},
  \citenamefont {Saltelli},\ and\ \citenamefont {Tarantola}}]{Chan:97}%
  \BibitemOpen
  \bibfield  {author} {\bibinfo {author} {\bibfnamefont {K.}~\bibnamefont
  {Chan}}, \bibinfo {author} {\bibfnamefont {A.}~\bibnamefont {Saltelli}}, \
  and\ \bibinfo {author} {\bibfnamefont {S.}~\bibnamefont {Tarantola}},\
  }\bibfield  {title} {\enquote {\bibinfo {title} {Sensitivity analysis of
  model output: variance-based methods make the difference},}\ }\href@noop {}
  {\bibfield  {journal} {\bibinfo  {journal} {Proc. of the 29th conf. on Winter
  simulation},\ \bibinfo {pages} {261--268}} (\bibinfo {year}
  {1997})}\BibitemShut {NoStop}%
\bibitem [{\citenamefont {Saltelli}\ \emph {et~al.}(2008)\citenamefont
  {Saltelli}, \citenamefont {Ratto}, \citenamefont {Andres}, \citenamefont
  {Campolongo}, \citenamefont {Cariboni}, \citenamefont {Gatelli},
  \citenamefont {Saisana},\ and\ \citenamefont {Tarantola}}]{Saltelli:08}%
  \BibitemOpen
  \bibfield  {author} {\bibinfo {author} {\bibfnamefont {A.}~\bibnamefont
  {Saltelli}}, \bibinfo {author} {\bibfnamefont {M.}~\bibnamefont {Ratto}},
  \bibinfo {author} {\bibfnamefont {T.}~\bibnamefont {Andres}}, \bibinfo
  {author} {\bibfnamefont {F.}~\bibnamefont {Campolongo}}, \bibinfo {author}
  {\bibfnamefont {J.}~\bibnamefont {Cariboni}}, \bibinfo {author}
  {\bibfnamefont {D.}~\bibnamefont {Gatelli}}, \bibinfo {author} {\bibfnamefont
  {M.}~\bibnamefont {Saisana}}, \ and\ \bibinfo {author} {\bibfnamefont
  {S.}~\bibnamefont {Tarantola}},\ }\href@noop {} {\emph {\bibinfo {title}
  {Global Sensitivity Analysis. {T}he Primer}}}\ (\bibinfo  {publisher}
  {Wiley},\ \bibinfo {year} {2008})\BibitemShut {NoStop}%
\end{thebibliography}

%

\end{document}